\documentclass[preprint,12pt]{elsarticle}

\usepackage[hmargin=1.3in,vmargin=1.3in]{geometry}
\usepackage{bbding}
\usepackage{mathrsfs}
\usepackage{}
\usepackage{amsfonts}
\usepackage{multirow}
\usepackage{amsfonts,amssymb,amsmath,amsthm,bm}
\usepackage[T1]{fontenc}
\usepackage[utf8]{inputenc}
\usepackage[colorlinks,
            linkcolor=blue,
            anchorcolor=blue,
            citecolor=blue]{hyperref}
\usepackage{graphicx}
\usepackage{booktabs}
\usepackage{longtable}
\usepackage{caption}
\usepackage{romanbar}
\usepackage{tikz}
\usepackage{pstricks}
\usepackage{graphicx} 
\usepackage[justification=centering]{caption} 
\usepackage{paralist}
\usepackage{array}

\newtheorem{theorem}{Theorem}
\newtheorem{example}{Example}

\newtheorem{remark}{Remark}
\newtheorem{corollary}{Corollary}
\newtheorem{lemma}{Lemma}

\journal{Finite Fields and Their Applications }
\begin{document}
\begin{frontmatter}
\title{Construction of MDS Euclidean Self-Dual Codes via Multiple Subsets}
\author[1]{Weirong Meng}
\ead{mwr@mail.nankai.edu.cn}
\author[2,3,4]{Weijun Fang}
\ead{fwj@sdu.edu.cn}
\author[1]{Fang-Wei Fu}
\ead{fwfu@nankai.edu.cn}
\author[5]{Haiyan Zhou\corref{cor}}
\ead{zhouhy@njnu.edu.cn}
\author[5]{Ziyi Gu}
\ead{3310284240@qq.com}
{\cortext[cor]{Corresponding author.}}
\tnotetext[label1]
{Submitted for possible publication.}
\address[1]{Chern Institute of Mathematics and LPMC, Nankai University, Tianjin, 300071, China}
\address[2]{State Key Laboratory of Cryptography and Digital Economy Security, Shandong University, Qingdao, 266237, China\\}
\address[3]{Key Laboratory of Cryptologic Technology and Information Security, Ministry of Education, Shandong University, Qingdao, 266237, China}
\address[4]{School of Cyber Science and Technology, Shandong University, Qingdao, 266237, China}
\address[5]{School of Mathematical Sciences and Institute of Mathematics, Nanjing Normal University, Nanjing, 210023, China}

\begin{abstract}
MDS self-dual codes have good algebraic structure, and their parameters are completely determined by the code length. In recent years, the construction of  MDS Euclidean self-dual codes with new lengths has become an important issue in coding theory. In this paper, we are committed to constructing new MDS Euclidean self-dual codes via generalized Reed-Solomon (GRS) codes and their extended (EGRS) codes. The main effort of our constructions is to find suitable subsets of finite fields as the evaluation sets, ensuring that the corresponding (extended) GRS codes are Euclidean self-dual. {Firstly}, we present a method for selecting evaluation sets from multiple intersecting subsets and provide a theorem to guarantee that the chosen evaluation sets meet the desired criteria. Secondly, based on this theorem, we construct six new classes of MDS Euclidean self-dual codes using the norm function, as well as the union of three multiplicity subgroups and their cosets respectively. Finally, in our constructions, the proportion of possible MDS Euclidean self-dual codes exceeds $85\%$, which is much higher than previously reported results.
\end{abstract}
\end{frontmatter}
\section{Introduction}
Both MDS codes and Euclidean self-dual codes play a significant role in coding theory due to their nice properties and wide applications.
Let $p$ be a prime and $q=p^m$ be a prime power. Let $\mathbb{F}_q$ be the finite field with $q$ elements. A $q$-ary $[n,k,d]$-linear code $C$ is defined as a linear $\mathbb{F}_q$-subspace of $\mathbb{F}_q^n$ with dimension $k$ and minimum Hamming distance $d$. The parameters of $C$ satisfy the following inequality,
$$d\leq n-k+1,$$
which is the well-known Singleton bound.
If $d=n-k+1$, then $C$ is called the \textit{maximum distance separable} (MDS) code. MDS codes are closely related to orthogonal arrays in combinatorial design and n-arcs in finite geometry \cite{MS77}. Furthermore, MDS codes also have widespread applications in distributed storage systems. 

For any two vectors $\textbf{x}=(x_1,x_2,\ldots, x_n)\in\mathbb{F}_q^n$ and  $\textbf{y}=(y_1,y_2,\ldots, y_n)\in\mathbb{F}_q^n$, define the Euclidean inner product 
$\langle\textbf{x},\textbf{y}\rangle=\sum\limits_{i=1}^{n}x_iy_i$. The Euclidean dual code $C^{\perp}$ of $C$ is defined by
$$C^{\perp}\triangleq\{\textbf{x}\in \mathbb{F}_q^n: \langle\textbf{x},\textbf{y}\rangle=0, \textnormal{for any}\  \textbf{y}\in C\}.$$
If $C\subseteq C^{\perp}$, then $C$ is called a Euclidean \textit{self-orthogonal} code. If $C=C^{\perp}$, then we call $C$ a Euclidean \textit{self-dual} code. Clearly, the length of a Euclidean self-dual code is even. 
Self-dual codes find applications in various fields, including coding theory, cryptography, and beyond. 
In \cite{C08} and \cite{DMS08}, the authors established a connection between self-dual codes and linear secret sharing schemes. Furthermore, Euclidean self-dual codes are closely linked to combinatorics \cite{HP03} and unimodular integer lattices \cite{CS99}. Given these connections, we focus on the intersection of these two classes of codes, namely, MDS Euclidean self-dual codes, a topic that has garnered significant attention over the past decade (see \cite{GKL08},\cite{BGGHK03},\cite{GK02},\cite{KL04},\cite{KL04-1},\cite{HK06},\cite{YC15}).

It is clear that the parameters of an MDS Euclidean self-dual code are completely determined by its length. Thus, we only need to consider the problem for which lengths an MDS Euclidean self-dual code exists. When $q$ is even, Grassl and Gulliver \cite{GG08} proved that there is a $q$-ary MDS Euclidean self-dual code of even length $n$ for all $n\leq q$. In \cite{G12} and \cite{TW17} , the authors constructed some new MDS Euclidean self-dual codes using cyclic codes and constacyclic codes. A criterion for generalized Reed-Solomon (GRS) codes to be MDS self-dual codes was given by Jin and Xing in \cite{JX17}. Yan \cite{Y18}, Fang and Fu \cite{FF19} developed this method to extended GRS codes. From then on, Zhang and Feng \cite{ZF20} introduced a unified approach to construct MDS Euclidean self-dual codes via GRS codes. By selecting suitable evaluation sets for (extended) GRS codes, MDS Euclidean self-dual codes are obtained in some ways. Specifically, in \cite{FLLL20} and \cite{LLL19}, the authors obtained some new MDS Euclidean self-dual codes via multiplicative subgroups. Fang \emph{et al.} \cite{FZXF22} took a subgroup of a finite field as the evaluation set, where its cosets are in a biggger subgroup. In \cite{FLL21}, the authors considered the evaluation sets as the union of two disjoint multiplicative subgroups and their cosets. Inspired by \cite{FLL21}, Huang \emph{et al.} further constructed some new families of MDS Euclidean self-dual codes in \cite{HFF21}. After this, Wan \emph{et al.} \cite{WLZ23} further generalized the above results and constructed six new classes of MDS Euclidean self-dual codes. And in \cite{FXF21}, Fang \emph{et al.} considered the union of two multiplicative subgroups with nonempty intersections and took their cosets as the evaluation sets. In \cite{XFL21}, Xie \emph{et al.} constructed some long MDS Euclidean self-dual codes by making use of short codes. And in \cite{ZF22}, Zhang and Feng presented some new constructions of MDS Euclidean self-dual codes via cyclotomy. Niu \emph{et al.} constructed two classes of MDS Hermitian self-dual codes in \cite{NYWH19}. Due to the good properties of the twisted GRS codes, the authors constructed some MDS, NMDS or 2-MDS Euclidean self-dual codes using twisted GRS codes in \cite{HYNL21},\cite{SYLH22} and \cite{SYS23}. In \cite{GLLS23}, Guo \emph{et al.} obtained several classes of Hermitian self-dual MDS and NMDS codes via twisted GRS codes. And in \cite{ZZT22,ZL24,ZL24-2,ZL24-3}, the authors presented some MDS or NMDS Euclidean self-dual codes derived from special twisted GRS codes. Additionally, Jin and Kan constructed NMDS Euclidean self-dual codes by making use of the properties of elliptic curves in \cite{JK19}. In \cite{HZ24}, Han and Zhang constructed some NMDS Euclidean self-dual codes via linear codes. In Table I of Section VI, we list some known results about the construction of MDS Euclidean self-dual codes. Despite the large amount of work has been done, constructing MDS Euclidean self-dual codes for all possible code lengths is still very challenging.

In this paper, inspired by the approach in \cite{FXF21}, we further explore the construction of MDS Euclidean self-dual codes via (extended) GRS codes. In \cite{ZF20}, Zhang and Feng gave the sufficient conditions for (extended) GRS codes to be self-dual codes with respect to the evaluation set (see Lemmas 2 and 3). Thus, the main idea of our constructions is to choose suitable evaluation sets such that the conditions in Lemma 2 or Lemma 3 hold.

Firstly, we present a method for selecting an evaluation set from multiple intersecting subsets of finite fields. Additionally, we provide a key result (see Theorem 1) that ensures the chosen evaluation set satisfies the conditions in Lemma 2 or Lemma 3. Secondly, based on Theorem 1, we construct six new classes of MDS Euclidean self-dual codes using three subsets of finite fields (see Theorems 3, 4, 5, 6 and Corollaries 1, 3). All codes derived from our constructions are summarized in Table 2 of Section 6. 
Finally, we make a detailed comparison of our constructions with previous results. Our approach achieves a proportion exceeding $85\%$ of all possible lengths of 
$q$-ary MDS Euclidean self-dual codes (see Table 3), representing the highest known ratio to date.

The organization of this paper is presented as follows. In Section 2, we introduce some basic notations and results about GRS codes and extended GRS codes. In Section 3, we present a method for choosing evaluation sets from multiple intersecting subsets of finite fields and propose a theorem to guarantee that the above evaluation sets can be taken as the desired evaluation sets. In Section 4, we provide some constructions of MDS Euclidean self-dual codes via the norm function. In Section 5, based on the union of three multiplicity subgroups and their cosets, several new classes of MDS Euclidean self-dual codes are also obtained. In Section 6, we make a comparison of our results with the known results. Finally, Section 7 concludes this paper.

\section{Preliminaries}
In this section, we state some notations and basic results on MDS Euclidean self-dual codes. 
 
Let $q$ be a prime power and $\mathbb{F}_q$ be the finite field with $q$ elements. Choose $n$ distinct elements $a_1, a_2, \ldots , a_n$ and $n$ nonzero elements $v_1, v_2, \ldots, v_n$ of  $\mathbb{F}_q$. Let $\mathbf{a}=(a_1,a_2,\ldots, a_n)$ and $\mathbf{v}=(v_1,v_2,\ldots, v_n)$. The $k$-dimensional generalized Reed-Solomon (GRS) code associated to $\mathbf{a}$ and $\mathbf{v}$ is defined as follows
 $$GRS_k(\mathbf{a},\mathbf{v})\triangleq\{(v_1f(a_1),\ldots, v_nf(a_n)): f(x)\in\mathbb{F}_q[x], \ \deg(f(x))\leq k-1\}.$$
 Furthermore, the $k$-dimensional extended GRS code associated to $\mathbf{a}$ and $\mathbf{v}$ is given as 
$$GRS_k(\mathbf{a},\mathbf{v},\infty)\triangleq\{(v_1f(a_1),\ldots, v_nf(a_n),f_{k-1}):f(x)\in\mathbb{F}_q[x],\ \deg(f(x))\leq k-1\},$$
where $f_{k-1}$ is the coefficient of $x^{k-1}$ in $f(x)$. The set $\{a_1,a_2,\dots,a_n\}$ is called the evaluation set of (extended) GRS code. And it is well-known that (extended) GRS codes and their dual codes are MDS codes.

For any subset $E\subseteq \mathbb{F}_q$, we define the polynomial $\pi_{E}(x)\in \mathbb{F}_q[x]$ as follows
\begin{equation*}
    \pi_{E}(x)\triangleq \prod_{e\in E}(x-e).
\end{equation*}
For any element $e\in E$, define 
\begin{equation*}
    \delta_{E}(e)\ \triangleq \prod_{e^{\prime} \in E, \ e^{\prime}\neq e}(e-e^{\prime}),
\end{equation*}
then we can easily verify that $\delta_{E}(e)=\pi^{\prime}_{E}(e)$,
where $\pi^{\prime}_{E}(x)$ is the formal derivative of $\pi_{E}(x)$.
When the set $E$ is represented as a disjoint union of $n$ sets, the following lemma gives the formula for calculating $\delta_{E}(e)$.
\begin{lemma}[\cite{FXF21}]\label{lem1}   
Let $E=\bigcup\limits_{i=1}^{n}E_{i}$, where $E_{1},E_{2},\dots, E_{n}$ are $n$ pairwise disjoint subsets of $E$, then for any $e\in E_i$, we have
 $$\delta_{E}(e)=\delta_{E_{i}}(e)\prod_{1 \leq j \leq n, j \neq i}\pi_{E_{j}}(e).$$
\end{lemma}
Let $\eta(x)$ be the quadratic character of $\mathbb{F}_q^{*}$, that is $\eta(x)=1$ if $x$ is a square in $\mathbb{F}_q^{*}$ and $\eta(x)=-1$ if $x$ is a non-square in $\mathbb{F}_q^{*}$. The following two lemmas provide sufficient conditions for (extended) GRS codes to be Euclidean self-dual codes, which ensure the existence of MDS Euclidean self-dual codes.
\begin{lemma}[\cite{ZF20}]\label{lem2}
Let $n$ be even. If there exits a subset $A \subseteq \mathbb{F}_{q}$ of size $n$, such that $\eta\big(\delta_{A}(a)\big)$ are the same for all $a \in A$, then there exists a $q$-ary MDS Euclidean self-dual code of length $n$.
\end{lemma}
\begin{lemma}[\cite{ZF20}]\label{lem3}
Let $n$ be odd. If there exits a subset $A \subseteq \mathbb{F}_{q}$ of size $n$, such that $\eta\big(-\delta_{A}(a)\big)=1$ for all $a \in A$, then there exists a $q$-ary MDS Euclidean self-dual code of length $n+1$.
\end{lemma}
\section{Evaluation sets from multiple sets of finite fields}

In this section, we will study the selection of evaluation sets for GRS codes from multiple intersecting subsets and provide a sufficient condition to guarantee that the evaluation sets chosen above satisfy the condition in Lemma \ref{lem2} or Lemma \ref{lem3}. 

Let $A_1,A_2,\dots,A_n$ be $n$ sets and 
$[n]=\{1,2,\dots,n\}$. Assume that $I$ is a subset of $[n]$ and  $\overline{I}=[n] \backslash I$, if $I\neq\emptyset$, then we define 
$$A_I\triangleq\bigcap\limits_{i\in I}A_i, \ \ \widetilde{A}_{{I}}\triangleq\bigcup\limits_{i\in I}A_i.$$ 
If $I=\emptyset$, then we define $A_I=\widetilde{A}_I=\emptyset$.  
Based on the definitions of $A_I$ and $\widetilde{A}_{{I}}$, we can divide the set $\bigcup\limits_{i=1}^n A_i$  into a union of pairwise disjoint sets, i.e., $\bigcup\limits_{i=1}^n A_i=\bigcup\limits_{ i=1}^{n}\left(\bigcup\limits_{I\subseteq [n],|I|=i}(A_I\backslash \widetilde{A}_{\overline{I}})\right)$. 
Choose 
\begin{equation}
S=\bigcup\limits_{k=0}^{\lfloor\frac{n-1}{2}\rfloor}\left(\bigcup\limits_{I\subseteq [n],|I|=2k+1}(A_I\backslash \widetilde{A}_{\overline{I}})\right).  \label{eq1}   
\end{equation}

Clearly, $S$ is a subset of $\bigcup\limits_{i=1}^n A_i$, which will be used to select evaluation sets from multiple intersecting subsets.
\begin{example} 
Let $n=2$, $A_1$ and $A_2$ be two sets. Then from Eq. (\ref{eq1}), we have 
 $$S=(A_1\backslash A_2)\bigcup(A_2\backslash A_1).$$
\end{example}
\begin{example}
Let $n=3$, $A_1$, $A_2$ and $A_3$ be three sets.  Then from the Eq. (\ref{eq1}), we have
$$S=(A_1\backslash(A_2\cup A_3))\bigcup(A_2\backslash(A_1\cup A_3))\bigcup(A_3\backslash(A_1\cup A_2))\bigcup(A_1\cap A_2\cap A_3).$$
In order to represent S more intuitively, the following Venn diagram is presented. When $n=2$ and $n=3$, the evaluation set S takes the shaded part of Figure 1 and Figure 2, respectively.

\begin{figure}[htbp]
\centering
\begin{minipage}[b]{0.3\textwidth}
    \centering
    \begin{tikzpicture}
    \draw[fill=none, draw=black] (0,0) circle(0.6);
    \draw[fill=none, draw=black] (0.6,0) circle(0.6);
    \begin{scope}
    \fill[gray!30] (0,0) circle(0.6); 
    \fill[gray!30] (0.6,0) circle(0.6);
    \end{scope}
    \begin{scope}
        \clip (0,0) circle(0.6);  
        \fill[white!30] (0.6,0) circle(0.6);  
    \end{scope}
    \draw[draw=black] (0,0) circle(0.6) node [left] {\scriptsize$A_1$}; 
    \draw[draw=black] (0.6,0) circle(0.6) node [right] {\scriptsize$A_2$};
    \end{tikzpicture}
    \caption{Two sets.} 
\end{minipage}%
\hspace{0.02\textwidth} 
\begin{minipage}[b]{0.5\textwidth}
    \centering
    \begin{tikzpicture}
    \draw[fill=none, draw=black] (0,0) circle(0.58) node {A\textsubscript{1}};
    \draw[fill=none, draw=black] (55:0.725) circle(0.58) node {A\textsubscript{2}};
    \draw[fill=none, draw=black] (0:0.75) circle(0.58) node {A\textsubscript{3}};
    \begin{scope}
    \fill[gray!30] (0,0) circle(0.58); 
    \fill[gray!30] (55:0.725) circle(0.58);
    \fill[gray!30] (0:0.75) circle(0.58);
    \end{scope}
    \begin{scope}
        \clip (0,0) circle(0.58);  
        \fill[white] (0:0.75) circle(0.58);  
    \end{scope}
    \begin{scope}
        \clip (55:0.725) circle(0.58);  
        \fill[white] (0:0.75) circle(0.58);  
    \end{scope}
    \begin{scope}
        \clip (55:0.725) circle(0.58);  
        \fill[white] (0,0) circle(0.58); 
    \end{scope}
    \begin{scope}
        \clip (55:0.725) circle(0.58);  
        \clip (0,0) circle(0.55); 
        \fill[gray!30] (0:0.75) circle(0.58); 
    \end{scope}
    \draw[draw=black] (0,0) circle(0.58) node [left] {\scriptsize$A_2$}; 
   \draw[draw=black] (55:0.725) circle(0.58) node [above] {\scriptsize$A_1$};
    \draw[draw=black] (0:0.75) circle(0.58) node [right] {\scriptsize$A_3$};
    \end{tikzpicture}
    \caption{Three sets.} 
\end{minipage}
\end{figure} 
\end{example}
If the set $S$ defined by Eq. (\ref{eq1}) satisfies the conditions in Lemma \ref{lem2} or Lemma \ref{lem3}, then the (extended) GRS code with $S$ as the evaluation set is a Euclidean self-dual code. Thus, our primary goal is to derive a sufficient condition to ensure $S$ satisfying 
Lemma \ref{lem2} or Lemma \ref{lem3}. We first give the following lemma as the preparation of the sufficient condition.

Let $A=\{a_1,a_2,\dots,a_s\}$ and $B=\{b_1,b_2,\dots,b_{s^{\prime}}\}$ be two sets, where $s, s^{\prime}\geq 1$. We define $A\uplus B=\{a_1,a_2,\dots,a_s,b_1,b_2,\dots,b_{s^{\prime}}\}$, allowing repeated elements to be included. It is clear that $A\uplus B$ may be a multiset in some cases. Therefore, we use $\widetilde{=}$ to represent the equality of multisets so as to distinguish it from the equality of ordinary sets. Let $B_l=\bigcup\limits_{j=1, \ j\neq l}^n\left(A_l\bigcap A_j\right)$, where $1\leq l\leq n$. 
\begin{lemma}\label{lem4} Based the above notations, the following results hold.

(1) For any $I\subseteq [n]$, if $|I|=1$, then 
\begin{equation}
(A_I\backslash \widetilde{A}_{\overline{I}})\cap B_l=\emptyset, \ \text{for~all}~ 1\leq l\leq n. \label{eq2}   
\end{equation}
If $2\leq |I|\leq n$, then
\begin{equation}
\#\{l \mid A_I\backslash \widetilde{A}_{\overline{I}}\subseteq B_l, 1\leq l\leq n\}=|I|.  \label{eq3}  
\end{equation} 

(2) Let $\mathcal{A}\widetilde{=}\biguplus\limits_{i=2}^{n}\biguplus\limits_{|I|=i}(\underbrace{A_I\backslash \widetilde{A}_{\overline{I}} \uplus \cdots\uplus A_I\backslash \widetilde{A}_{\overline{I}}}_{|I|})$ and $\mathcal{B}\widetilde{=}\biguplus\limits_{l=1}^n B_l$, 
then $\mathcal{A}\widetilde{=}\mathcal{B}$.
\end{lemma}
\begin{proof}
(1) Firstly, we consider the case of $|I|=1$. Without loss of generality, we can assume $I=\{1\}$. For any $ 1\leq l\leq n $, by the operation of sets,
 $$\left(A_I\backslash \widetilde{A}_{\overline{I}}\right)\bigcap B_l=\left(A_{1}\backslash (\bigcup\limits_{2\leq j\leq n}A_j)\right)\bigcap \left(\bigcup\limits_{j=1, \ j\neq l}^n(A_l\bigcap A_j)\right)=\emptyset.$$
Next, we consider the case of $2\leq |I|\leq n$. Without loss of generality, we can assume $I=\{1,2,\ldots,|I|\}$. On the one hand, for any $1\leq l\leq |I| $, it follows from $\bigcap\limits_{1\leq i\leq |I|}A_{i}\subseteq (A_l\cap(\bigcup\limits_{i=1, i\neq l}^n A_i))=B_l$ that
$$A_I\backslash \widetilde{A}_{\overline{I}}=\left((\bigcap\limits_{1\leq i\leq |I|}A_{i})\backslash (\bigcup\limits_{|I|+1\leq j\leq n}A_j)\right) \subseteq \left(\bigcup\limits_{j=1, \ j\neq l}^n(A_l\bigcap A_j)\right)= B_l.$$
On the other hand, for any $|I|+1\leq l\leq n$, we can get
 $\left(A_I\backslash \widetilde{A}_{\overline{I}}\right)\bigcap B_l=\emptyset$. Thus, for any $I\subseteq N$ with $ |I|\geq 2$, we obtain
$$\#\{l\ |A_I\backslash \widetilde{A}_{\overline{I}}\subseteq B_l, \ 1\leq l\leq n\}=|I|.$$

(2) For any $1\leq l\leq n$, by $B_l\subseteq \bigcup\limits_{i=1}^n A_i=\bigcup\limits_{ i=1}^{n}\left(\bigcup\limits_{|I|=i}(A_I\backslash \widetilde{A}_{\overline{I}})\right)$ and Eq. (\ref{eq2}),  we can deduce that $B_l$ is the union of some sets $A_I\backslash \widetilde{A}_{\overline{I}}$, where $|I|\geq 2$. Furthermore, it can be derived from Eq. (\ref{eq3}) that for each $I\subseteq [n]$ with $|I|\geq 2$, the number of sets $B_1, B_2,\cdots, B_n$ that contain set $A_I\backslash {\widetilde{A}_{\overline{I}}}$ is $|I|$. Thus,
$B_1\uplus\cdots\uplus B_n\widetilde{=}\biguplus\limits_{i=2}^n\biguplus\limits_{|I|=i}(\underbrace{A_I\backslash \widetilde{A}_{\overline{I}} \uplus \cdots\uplus A_I\backslash \widetilde{A}_{\overline{I}}}_{|I|})$, implying $\mathcal{A}=\mathcal{B}$.
\end{proof}
Based on Lemma \ref{lem4}, we provide a sufficient condition for $S$ to satisfy the conditions in Lemma \ref{lem2} or Lemma \ref{lem3}, which is the key theorem for the next constructions.
\begin{theorem}\label{th1}
Let $A_1, \ldots, A_n$ be $n$ subsets of $\mathbb{F}_q$.  For the fixed constant $c\in\{1,-1\}\subseteq \mathbb{F}_q$ and any $a\in A_I\backslash \widetilde{A}_{\overline{I}}$, where $I\subseteq [n]$ with $|I|=2k+1$ and $0\leq k\leq \lfloor\frac{n-1}{2}\rfloor$, if 
 $$\eta\left(\prod_{i\in I}\delta_{A_i}(a)\prod_{j\in \overline{I}}\pi_{A_j}(a)\right)=c,$$
 then for all $e\in S$,
 $$\eta(\delta_S(e))=c.$$
\end{theorem}
\begin{proof}
Let $e\in A_{I^{\prime}}\backslash \widetilde{A}_{\overline{I^{\prime}}}\subseteq S$, where $|I^{\prime}|=2k^{\prime}+1$ and $0\leq k^{\prime}\leq \lfloor\frac{n-1}{2}\rfloor$. 
When $|I^{\prime}|=1$, without loss of generality, we can assume $I^{\prime}=\{1\}$, then 
$A_{I^{\prime}}\backslash \widetilde{A}_{\overline{I^{\prime}}}=A_1\backslash(A_2\cup\cdots\cup A_n).$ By Lemma \ref{lem1},
 \begin{eqnarray}
     \delta_S(e)&=&\left(\delta_{ A_{I^{\prime}}\backslash \widetilde{A}_{\overline{I^{\prime}}}}(e)\prod_{|I|=1, I\neq I^{\prime}}\pi_{ A_I\backslash {\widetilde{A}_{\overline{I}}}}(e)\right)\prod_{k=1}^{\lfloor\frac{n-1}{2}\rfloor}\prod_{|I|=2k+1}\pi_{ A_{I}\backslash {\widetilde{A}_{\overline{I}}}}(e)\nonumber\\
     &=&\frac{1}{\prod\limits_{l=1}^n\pi_{B_l}(e)}\left(\delta_{A_1}(e)\prod_{j\in \overline{I^{\prime}}}\pi_{A_j}(e)\right)\left(\prod_{k=1}^{\lfloor\frac{n-1}{2}\rfloor}\prod_{|I|=2k+1}\pi_{ A_{I}\backslash {\widetilde{A}_{\overline{I}}}}(e)\right).\label{eq4}
 \end{eqnarray}
According to Lemma \ref{lem4}, 
\begin{equation}
 \eta\left(\prod_{l=1}^n\pi_{B_l}(e)\right)=\eta\left(\prod_{i=2}^n\prod_{|I|=i}(\pi_{A_I\backslash {\widetilde{A}_{\overline{I}}}}(e))^{|I|}\right)
      =\eta\left(\prod_{k=1}^{\lfloor\frac{n-1}{2}\rfloor}\prod_{|I|=2k+1}\pi_{ A_{I}\backslash \widetilde{A}_{\overline{I}}}(e)\right). \label{eq5}  
\end{equation} 
Hence, when $|I^{\prime}|=1$, by {Eqs.} (\ref{eq4}) and (\ref{eq5}), 
$\eta(\delta_S(e))=\eta(\delta_{A_1}(e)\prod_{j\in \overline{I^{\prime}}}\pi_{A_j}(e))$ can be obtained.

For any $1\leq k^{\prime}\leq \lfloor\frac{n-1}{2}\rfloor$ and $|I^{\prime}|=2k^{\prime}+1$, we can assume $I^{\prime}=\{1,2,\dots,2k^{\prime}+1\}$, then $$ A_{I^{\prime}}\backslash \widetilde{A}_{\overline{I^{\prime}}}=(A_1\cap A_2\cap\cdots\cap A_{2k^{\prime}+1})\backslash(A_{2k^{\prime}+2}\cup\cdots\cup A_n),$$
which implies that $e\in A_i$ for any $1\leq i\leq 2k^{\prime}+1$.
Since $|I^{\prime}|\geq 2$, it follows from Lemma \ref{lem4} that $e\in A_{I^{\prime}}\backslash \widetilde{A}_{\overline{I^{\prime}}}\subseteq B_l$ for any $1\leq l\leq 2k^{\prime}+1$.
Thus, by Lemma \ref{lem1},
\begin{eqnarray*}
    \delta_S(e)&=&\delta_{ A_{I^{\prime}}\backslash \widetilde{A}_{\overline{I^{\prime}}}}(e)\left(\prod_{|I|=1}\pi_{ A_I \backslash {\widetilde{A}_{\overline{I}}}}(e)\right)\left(\prod_{k=1}^{\lfloor\frac{n-1}{2}\rfloor}\prod_{|I|=2k+1, I\neq I^{\prime}}\pi_{ A_{I}\backslash {\widetilde{A}_{\overline{I}}}}(e)\right)\\
    &=&\delta_{ A_{I^{\prime}}\backslash \widetilde{A}_{\overline{I^{\prime}}}}(e)\left(\prod_{i=1}^{2k^{\prime}+1}\frac{\delta_{A_i}(e)}{\delta_{B_i}(e)}\prod_{j=2k^{\prime}+2}^n\frac{\pi_{A_j}(e)}{\pi_{B_j}(e)}\right)\left(\prod_{k=1}^{\lfloor\frac{n-1}{2}\rfloor}\prod_{|I|=2k+1, I\neq I^{\prime}}\pi_{ A_{I}\backslash {\widetilde{A}_{\overline{I}}}}(e)\right)\\
    &=&\frac{\delta_{ A_{I^{\prime}}\backslash \widetilde{A}_{\overline{I^{\prime}}}}(e)}{\prod\limits_{i\in I^{\prime}}\delta_{B_i}(e)\prod\limits_{j\in \overline{I^{\prime}}}\pi_{B_j}(e)}\left(\prod_{i\in I^{\prime}}\delta_{A_i}(e)\prod_{j\in \overline{I^{\prime}}}\pi_{A_j}(e)\right)\left(\prod_{k=1}^{\lfloor\frac{n-1}{2}\rfloor}\prod_{|I|=2k+1, I\neq I^{\prime}}\pi_{ A_{I}\backslash {\widetilde{A}_{\overline{I}}}}(e)\right).
\end{eqnarray*}
Similarly, from Lemma \ref{lem4}, 
$$\biguplus_{l=1}^n B_l \widetilde{=} (\underbrace {A_{I^{\prime}}\backslash \widetilde{A}_{\overline{I^{\prime}}}\uplus \cdots \uplus A_{I^{\prime}}\backslash \widetilde{A}_{\overline{I^{\prime}}}}_{|I^{\prime}|})\biguplus \left(\biguplus_{i=2}^n\biguplus_{|I|=i, I\neq I^{\prime}}(\underbrace{A_I\backslash \widetilde{A}_{\overline{I}} \uplus \cdots\uplus A_I\backslash \widetilde{A}_{\overline{I}}}_{|I|})\right).$$
Hence, when $e\in A_{I^{\prime}}\backslash \widetilde{A}_{\overline{I^{\prime}}}$, we can show that
\begin{eqnarray*}
  \eta\left(\prod_{i\in I^{\prime}}\delta_{B_i}(e)\prod_{j\in \overline{I^{\prime}}}\pi_{B_j}(e)\right)&=&\eta\left(\left(\delta_{ A_{I^{\prime}}\backslash \widetilde{A}_{\overline{I^{\prime}}}}(e)\right)^{2k^{\prime}+1}\cdot\ \prod_{i=2}^{n}\prod_{|I|=i, I\neq I^{\prime}}\left(\pi_{ A_{I}\backslash {\widetilde{A}_{\overline{I}}}}(e)\right)^{|I|}\right) \\
  &=&\eta\left(\delta_{ A_{I^{\prime}}\backslash \widetilde{A}_{\overline{I^{\prime}}}}(e)\cdot\ \prod_{k=1}^{\lfloor\frac{n-1}{2}\rfloor}\prod_{|I|=2k+1, I\neq I^{\prime}}\pi_{ A_{I}\backslash {\widetilde{A}_{\overline{I}}}}(e)\right),
\end{eqnarray*}
which implies that
$\eta(\delta_S(e))=\eta(\prod\limits_{i\in I^{\prime}}\delta_{A_i}(e)\prod\limits_{j\in \overline{I^{\prime}}}\pi_{A_j}(e))$. 
\end{proof}
\begin{remark}
(1) When $n=2$, $S=\bigcup\limits_{|I|=1}(A_I\backslash \widetilde{A}_{\overline{I}})=(A_1\backslash A_2)\bigcup (A_2\backslash A_1).$ For any $a_1\in A_1\backslash A_2$ and $a_2\in A_2 \backslash A_1$, if $$\eta(\delta_{A_1}(a_1)\pi_{A_2}(a_1))=\eta(\delta_{A_2}(a_2)\pi_{A_1}(a_2))=c,$$ 
then, according to Theorem \ref{th1}, $\eta(\delta_{S}(e))=c$ for any $e\in S$.
Clearly, Lemma 5 in  \cite{FXF21} is a special case of Theorem \ref{th1}.\\
(2) When $n=3$, $S=\left(A_1\backslash \left(A_2\cup A_3\right)\right)\bigcup \left(A_2\backslash \left(A_1\cup A_3\right)\right)\bigcup \left(A_3\backslash \left(A_2\cup A_1\right)\right) \bigcup (A_1\cap A_2$ $ \cap A_3).$
For any $a_1\in A_1\backslash \left(A_2\cup A_3\right)$, $a_2\in A_2\backslash (A_1 \cup A_3)$, $a_3\in A_3\backslash \left(A_1\cup A_2\right)$ and $a_4\in A_1\cap A_2\cap A_3$, if
\begin{equation*}
 \begin{split}
  \eta(\delta_{A_1}(a_1)\pi_{A_2}(a_1)\pi_{A_3}(a_1))
  &=\eta(\delta_{A_2}(a_2)\pi_{A_1}(a_2)\pi_{A_3}(a_2))\\
  &=\eta(\delta_{A_3}(a_3)\pi_{A_1}(a_3)\pi_{A_2}(a_3))\\
  &=\eta(\delta_{A_1}(a_4)\delta_{A_2}(a_4)\delta_{A_3}(a_4))
  =c,   
 \end{split}   
\end{equation*}
then, according to Theorem \ref{th1}, we have $\eta(\delta_{S}(e))=c$ for any $e\in S$.
\end{remark}

\section{MDS Euclidean self-dual codes from norm functions}
In this section, based on Lemma \ref{lem2} or Lemma \ref{lem3}, we will present some new constructions of MDS Euclidean self-dual codes via Theorem \ref{th1}. Throughout this section, let $q=r^2$ and $r=p^m$, where $p$ is an odd prime.

Suppose that $\theta$ is a primitive element of $\mathbb{F}_q$ and  $l$ is a positive integer with $l|(r-1)$. Let $\alpha=\theta^{\frac{q-1}{l}}$, $\beta=\theta^{r+1}$, $M =\langle \alpha \rangle$ and $\mathbb{F}_r^{*}=\langle \beta \rangle$. Clearly, $M$ is the subgroup of $\mathbb{F}^*_r$ of size $l$. 
The sufficient and necessary condition for two cosets $\beta^{i}M$ and $\beta^{j}M$ to be different is given as follows.
\begin{lemma}\label{lem5}
For any $0\leq i<j < r-1$,
$\beta^{i}M\neq \beta^{j}M$ holds if and only if $i\not\equiv  j\pmod {\frac{r-1}{l}}$ holds.
\end{lemma}
\begin{proof}
For any  $0\leq i<j < r-1$, we can show 
\begin{eqnarray*}
    \beta^{i}M=\beta^{j}M& \Leftrightarrow&\beta^{i-j}=\theta^{(r+1)(i-j)}\in M\\
    & \Leftrightarrow&\frac{q-1}{l}\mid(r+1)(i-j)\\
    & \Leftrightarrow&\frac{r-1}{l}\mid(i-j).
\end{eqnarray*}
\end{proof}

Suppose that $s$ is even with $0\leq s\leq \frac{r-1}{l}-1$. Let $b_0=1$, $b_i=\beta^{i}$, $b_{\frac{s}{2}+i}=b_{i}^{-1}$ for any $1\leq i\leq \frac{s}{2}$. By Lemma \ref{lem5}, $b_0,b_1,\dots,b_s$ are $s+1$ distinct representations of  $\mathbb{F}_r^{*}\big/M$.
For any $0\leq k\leq s$, define
\begin{equation}
 H_k \triangleq b_k M=\{b_k m\mid m\in M\}.\label{eq6}  
\end{equation}
Then $H_k$ is a subset of $\mathbb{F}_r^{*}$ with $|H_k|=l$, and $H_i\cap H_j=\emptyset$ for any $0\leq i\neq j\leq s$.

Let $N(x)=x^{\frac{q-1}{r-1}}=x^{r+1}$ be the norm function from $\mathbb{F}_q^{*}$ to $\mathbb{F}_r^{*}$. For any $0\leq i\leq l-1$, define 
\begin{equation}
 N_i \triangleq \{ x\in \mathbb{F}_q^{*}\mid N(x)=\alpha^i\}=\theta^{\frac{r-1}{l} i}\langle \theta^{r-1}\rangle.\label{eq7}
\end{equation}
Then $N_i$ is a subset of $\mathbb{F}_q^{*}$ with $|N_i|=r+1$, and $N_i\cap N_j=\emptyset$ for any $0\leq i\neq j\leq l-1$.
\begin{lemma}\label{lem6}
For any $0\leq i\leq l-1$, if $l$ is even and $r \equiv 3 \pmod 4$, then 
\begin{equation*}
N_i\cap\mathbb{F}_r^{*} =
  \begin{cases}
      \emptyset, & \text{if } i \ \text{is }\text{odd},\\
      \{ \theta^{\frac{i}{2}\cdot\frac{q-1}{l}},\theta^{\left(\frac{i}{2}+\frac{l}{2}\right)\cdot\frac{q-1}{l}}\}\subseteq H_0, & \text{if } i~ \text{is }\text{even}.
    \end{cases}
\end{equation*}
\end{lemma}
\begin{proof}
For any $0\leq i\leq l-1$, if $e\in N_i\cap \mathbb{F}_r^{*}$, then there exist some $0\leq i_1\leq r$ and $\ 0\leq i_2\leq r-2$ such that 
\begin{equation}
  e=\theta^{\frac{r-1}{l} \cdot \ i+i_1(r-1)}=\theta^{i_2(r+1)}.  \label{eq8} 
\end{equation}
It is clear that Eq. (\ref{eq8}) is equivalent to 
\begin{equation*}
 \theta^{\left(\frac{r-1}{l} i+i_1(r-1)-i_2(r+1)\right)}=1.  
\end{equation*}
Hence, $q-1\mid\left(\frac{r-1}{l} i+i_1(r-1)-i_2(r+1)\right)$, which implies that 
\begin{equation}
 r+1\bigg|\left(2i_1-\frac{r-1}{l} i\right).\label{eq9}
\end{equation}
Thus to obtain $|N_i\cap \mathbb{F}_r^{*}|$ for any $0\leq i\leq l-1$, we only need to determine the number of $i_1$ satisfying $0\leq i_1\leq r$ and Eq. (\ref{eq9}).
Since $l$ is even and $r \equiv 3 \pmod 4$, $\frac{r-1}{l}$  is odd. If $i$ is odd, then Eq. (\ref{eq9}) does not hold, which implies that $N_i\cap\mathbb{F}_r^{*} =\emptyset$. If $i$ is even, by $0\leq 2i_1\leq 2r$, $0\leq i\leq l-1\leq r-2$, 
 $$-(r+1)< -\frac{r-1}{l}(l-1)\leq 2i_1-\frac{r-1}{l}i\leq 2r<2(r+1).$$
Therefore, Eq. (\ref{eq9}) holds if and only if $2i_1-\frac{r-1}{l}i=0\cdot (r+1)$ or $2i_1-\frac{r-1}{l}i=1\cdot (r+1)$, i.e., $i_1=\frac{i}{2}\cdot\frac{r-1}{l}$ or $i_1=\frac{i}{2}\cdot\frac{r-1}{l}+\frac{r+1}{2}$, which implies that
 $N_i\cap\mathbb{F}_r^{*} =\{ \theta^{\frac{i}{2}\cdot \frac{q-1}{l}}, \theta^{(\frac{i}{2}+\frac{l}{2})\cdot \frac{q-1}{l}}\}\subseteq H_0$.
\end{proof}
To facilitate subsequent calculations, we give the following lemmas.
\begin{lemma}\label{lem7}
With the above notations, the following results hold.

(1) For any $ 0\leq i\leq l-1$, we have $\pi_{N_i}(x)=\prod\limits_{a\in N_i}(x-a)=N(x)-\alpha^i=x^{r+1}-\alpha^{i}$.   

(2) For any $ 0\leq i\leq l-1$, we have $\delta_{N_i}(x)=(r+1)x^r$.
\end{lemma}
\begin{proof}
(1) On the one hand, $\pi_{N_i}(x)=\prod\limits_{a\in N_i}(x-a)$ and $\deg (\pi_{N_i}(x))=|N_i|=r+1$. On the other hand, $N(a)-\alpha^{i}=0$ for any $a\in N_i$ and $\deg (N(x)-\alpha^{i})=r+1$, which implies that $\pi_{N_i}(x)=N(x)-\alpha^{i}=x^{r+1}-\alpha^{i}$.  \\
(2) According to the definitions of $\pi_{N_i}(x)$ and $\delta_{N_i}(x)$, the result is clear.
\end{proof}

\begin{lemma}\label{lem8}
Let $q=r^2$ and $r=p^m~(p\geq 3)$. Suppose that $l$ and $s$ are even with $l|(r-1)$, where $0\leq s\leq \frac{r-1}{l}-1$. Let $H_k$ and $N_i$ be defined in Eqs. (\ref{eq6}) and (\ref{eq7}) respectively, where $0\leq k\leq s$, $0\leq i\leq l-1$. For any $b\in N_i$ and {$b\notin H_0$}, we have
$$\eta\left(\pi_{H_0}(b)\right)=1, \ \eta\left(\prod_{k=1}^{s}\pi_{H_k}(b) \right)=1.$$
\end{lemma}
\begin{proof}
Firstly, we consider $\pi_{H_0}(b)$. Since $b\in N_i$ and $H_0\subseteq \mathbb{F}_r^*$,  
\begin{equation}
\left(\pi_{H_0}(b)\right)^r=\prod_{\xi\in H_0}(b^r-\xi)=\prod_{\xi\in H_0}(b^{-1}\alpha^{i}-\xi)=\prod_{\xi\in H_0}(b^{-1}-\alpha^{-i} \xi). \label{eq10}   
\end{equation}
Note that $\alpha^{-i}\xi$ runs over $H_0$ when $\xi$ runs over $H_0$, thus
\begin{equation}
\left(\pi_{H_0}(b)\right)^r=
\prod_{\xi\in H_0}(b^{-1}-\xi). \label{eq11}   
\end{equation}
When $\xi$ runs over $H_0$, we can also know that $\xi^{-1}$ runs over $H_0$ , by Eq. (\ref{eq11}), we have 
$$\left(\pi_{H_0}(b)\right)^{-1}=\left(\pi_{H_0}(b)\right)^{-r^2}=\left(\prod\limits_{\xi\in H_0}(b^{-1}-\xi)\right)^{-r}=\prod\limits_{\xi\in H_0}(b^{r}-\xi).$$ Similar to the calculation of Eq. (\ref{eq10}), 
$$\left(\pi_{H_0}(b)\right)^{-1}=\prod\limits_{\xi\in H_0}(b^{r}-\xi)=\prod\limits_{\xi\in H_0}(b^{-1}-\xi)$$ can be obtained.
Thus, 
$$\left(\pi_{H_0}(b)\right)^{-1}=\prod_{\xi\in H_0}(b^{-1}-\xi)=     \left(\pi_{H_0}(b)\right)^r,$$ which implies that $\left(\pi_{H_0}(b)\right)^{r+1}=1$.
Since $\pi_{H_0}(b)\in \mathbb{F}_q^{*}=\langle \theta\rangle$, there exists an integer $x$ such that $\pi_{H_0}(b)=\theta^{(r-1)\cdot x}$, so $\pi_{H_0}(b)$ is a square element in $\mathbb{F}_q$. 

Next, we consider $\prod\limits_{k=1}^{s}\pi_{H_k}(b) $. Note that $\alpha^{-i}\xi$ runs over $H_k$ when $\xi$ runs over $H_k$,
\begin{equation}
\left(\pi_{H_k}(b)\right)^r=\prod_{\xi\in H_k}(b^{-1}-\xi) \label{eq12}  
\end{equation}  can be obtained similarly.
By $b\in N_i$ and Eq. (\ref{eq12}), 
$$\left(\pi_{H_k}(b)\right)^{-1}=\left(\prod\limits_{\xi\in H_k}(b^{-1}-\xi)\right)^{-r}=\prod\limits_{\xi\in H_{k}}(b^{-1}-\alpha^{-i}\xi^{-1}).$$
Since $b_{\frac{s}{2}+i}=b_i^{-1}$, we can show that $\alpha^{-i}\xi^{-1}$runs over $H_{k+\frac{s}{2}\pmod s}$ when $\xi$ runs over $H_k$, which implies that 
$$\left(\pi_{H_k}(b)\right)^{-1}=\prod\limits_{\xi\in H_{k+\frac{s}{2}\pmod s}}(b^{-1}-\xi).$$
Hence, $$ \left(\prod_{k=1}^s\pi_{H_k}(b)\right)^{-1}=\prod_{k=1}^s\prod_{\xi\in H_k}(b^{-1}-\xi)=\left(\prod_{k=1}^s\pi_{H_k}(b)\right)^r,$$ i.e.,  $\left(\prod\limits_{k=1}^s\pi_{H_k}(b)\right)^{r+1}=1$.
Since $\prod\limits_{k=1}^s\pi_{H_k}(b)\in \mathbb{F}_q^{*}=\langle \theta\rangle$, there exists an integer $x$ such that $\prod\limits_{k=1}^s\pi_{H_k}(b)=\theta^{(r-1)\cdot x}$, implying that $\prod\limits_{k=1}^s\pi_{H_k}(b)$ is a square element in $\mathbb{F}_q$.   
\end{proof}
Recall that $q=r^2$ and $r=p^m$, $M =\langle \alpha \rangle$ is the subgroup of $\mathbb{F}_r^{*}$ of order $l$. When $l$ is even, according to Theorem \ref{th1} and Lemma \ref{lem6}, we provide the following construction based on the sets given in Eq. (\ref{eq6}) and Eq. (\ref{eq7}).
\begin{theorem}\label{th2}
 Let $q=r^2$ and $r=p^m$ with $r\equiv3 \pmod4$. Suppose that $l$ and $s$ are even with $l|(r-1)$, $0\leq s\leq \frac{r-1}{l}-1$. For any $0\leq l_1\leq \frac{l}{2}$, $0\leq l_2\leq \frac{l}{2}$, put $n=sl+(l_1+l_2)(r+1)+1$, then there exists a $q$-ary MDS Euclidean self-dual code of length $n+1$.  
\end{theorem}
\begin{proof}
Let $$A=\bigcup\limits_{k=1}^{s}H_k, B=\bigcup\limits_{i=0}^{l_1-1} N_{2i+1}, C=\bigcup\limits_{j=0}^{l_2-1}N_{2j}, B^{\prime}=B\cup\{0\},$$ where $H_k$ and $N_i$ are defined in Eq. (\ref{eq6}) and Eq. (\ref{eq7}), respectively. Choose 
$$S=(A\backslash (B^{\prime}\cup C))\bigcup (B^{\prime}\backslash \left(A\cup C\right))\bigcup (C\backslash (A\cup B^{\prime}))\bigcup(A\cap B^{\prime}\cap C),$$   
then by Lemma \ref{lem6}, 
$$|S|=|A|+|B^{\prime}|+|C|=sl+(l_1+l_2)(r+1)+1=n.$$

Firstly, for any $a\in A \backslash (B^{\prime}\cup C)$, we will calculate $\eta(\delta_{A}(a)\pi_{B^{\prime}}(a)\pi_{C}(a))$.
Assume that $a\in H_k$, where $1\leq k \leq s$, then by Lemma \ref{lem1} and Lemma \ref{lem7},
\begin{eqnarray*}
     \delta_{A}(a)\pi_{B^{\prime}}(a)\pi_{C}(a)&=& \pi_{H_k}^{\prime}(a)\prod_{m=1, m\neq k}^{s}\pi_{H_m}(a)\left(\prod_{i=0}^{l_1-1}\pi_{N_{2i+1}}(a)\cdot a\right)\prod_{j=0}^{l_2-1}\pi_{N_{2j}}(a)\\
     &=&a \cdot \pi_{H_k}^{\prime}(a)\prod_{m=1, m\neq k}^{s}\pi_{H_m}(a)\prod_{i=0}^{l_1-1}(N(a)-\alpha^{2i+1})\prod_{j=0}^{l_2-1}(N(a)-\alpha^{2j}).
\end{eqnarray*}
Since $H_k\subseteq \mathbb{F}_r^{*}$ and $M\subseteq \mathbb{F}_r^{*}$, we  can ontain $\pi_{H_k}^{\prime}(a)$, $ \pi_{H_k}(a)$ and $N(a)\in \mathbb{F}_r^{*}$, which implies that $\delta_{A}(a)\pi_{B^{\prime}}(a)\pi_{C}(a)\in \mathbb{F}_r.$ Thus, $$\eta\left(\delta_{A}(a)\pi_{B^{\prime}}(a)\pi_{C}(a)\right)=1$$ since each element of $\mathbb{F}_r$ is a square element in $\mathbb{F}_q$.

Secondly, for any $b\in B^{\prime} \backslash (A\cup C)$, we will calculate $\eta(\pi_{A}(b)\delta_{B^{\prime}}(b)\pi_{C}(b))$. If $b=0$, then by Lemma \ref{lem1} and Lemma \ref{lem7},
\begin{eqnarray*}
     \pi_{A}(0)\delta_{B^{\prime}}(0)\pi_{C}(0)&=& \prod_{k=1}^{s}\pi_{H_k}(0)\prod_{i=0}^{l_1-1}\pi_{N_{2i+1}}(0)\prod_{j=0}^{l_2-1}\pi_{N_{2j}}(0)\\
     &=&(-1)^{l_1+l_2}\prod_{k=1}^{s}\pi_{H_k}(0)\prod_{i=0}^{l_1-1}\prod_{j=0}^{l_2-1}\alpha^{2i+2j+1}.
\end{eqnarray*}
Note that $q=r^2$ and $-1=\theta^{\frac{q-1}{2}}$, thus $\eta(-1)=1$. Since $H_k\subseteq \mathbb{F}_r^{*}$ and $0\in \mathbb{F}_r$, $\eta\left(\prod\limits_{k=1}^{s}\pi_{H_k}(0)\right)=1$. Moreover, by $\alpha=\theta^{\frac{q-1}{l}}$ and $l\mid (r-1)$, we can get $\eta(\alpha)=1$. In a word, 
\begin{equation}
\eta\left(\pi_{A}(0)\delta_{B^{\prime}}(0)\pi_{C}(0)\right)=1.    
\end{equation}
If $b \neq 0$, then we assume  $b\in N_{2i+1}$, where $0\leq i\leq l_1-1$. Then by Lemma \ref{lem1} and Lemma \ref{lem7},
 \begin{eqnarray*}
     \pi_{A}(b)\delta_{B^{\prime}}(b)\pi_{C}(b)&=& \prod_{k=1}^{s}\pi_{H_k}(b)\left(b\cdot\pi_{N_{2i+1}}^{\prime}(b)\prod_{ i_1=0, i_1\neq i}^{l_1-1}\pi_{N_{2i_1+1}}(b)\right) \prod_{j=0}^{l_2-1}\pi_{N_{2j}}(b)\\
      &=&b^{r+1}\prod_{i_1=0, i_1\neq i}^{l_1-1}(N(b)-\alpha^{2i_{1}+1})\prod_{j=0}^{l_2-1}(N(b)-\alpha^{2j})\prod_{k=1}^{s}\pi_{H_k}(b).
\end{eqnarray*}
Since $r+1$ is even and $\prod\limits_{i_1=0, i_1\neq i}^{l_1-1}(N(b)-\alpha^{2i_1+1})\prod\limits_{j=0}^{l_2-1}(N(b)-\alpha^{2j})\in \mathbb{F}_r^{*}$, it is sufficient to consider $\eta\left(\prod\limits_{k=1}^{s}\pi_{H_k}(b)\right)$. By Lemma \ref{lem8}, we can show that 
\begin{equation}
\eta\left(\pi_{A}(b)\delta_{B^{\prime}}(b)\pi_{C}(b)\right)=\eta\left(\prod_{k=1}^{s}\pi_{H_k}(b) \right)=1.  \label{eq14}  
\end{equation}
Finally, for any $c\in C \backslash (A\cup B^{\prime})$, we assume  $c\in N_{2j}$, where $0\leq j\leq l_2-1$. Similarly, we have
$$\eta\left(\pi_A(c)\pi_{B^{\prime}}(c)\delta_C(c)\right)=1.$$

In summary, for any $a\in A \backslash (B^{\prime}\cup C)$, $b\in B^{\prime} \backslash (A\cup C)$ and $c\in C \backslash (A\cup B^{\prime})$, we have
$$\eta\left(\delta_{A}(a)\pi_{B^{\prime}}(a)\pi_{C}(a)\right)=\eta\left(\pi_{A}(b)\delta_{B^{\prime}}(b)\pi_{C}(b)\right)=\eta\left(\pi_A(c)\pi_{B^{\prime}}(c)\delta_C(c)\right)=1.$$   
According to Theorem \ref{th1}, for any $e\in S$, we can get $\eta\left(-\delta_S(e)\right)=\eta\left(\delta_S(e)\right)=1$. Again by Lemma \ref{lem3}, there exists a $q$-ary MDS Euclidean self-dual code of length $n+1$.  
\end{proof}
By modifying the evaluation set of Theorem \ref{th2}, we can provide another construction of MDS Euclidean self-dual codes.
\begin{theorem}\label{th3}
Let $q=r^2$ and $r=p^m$ with $r\equiv3 \pmod 4$. Suppose that $l$ and $s$ are even with $l|(r-1)$, $0\leq s\leq \frac{r-1}{l}-1$. For any $0\leq l_1\leq \frac{l}{2}$, $0\leq l_2\leq \frac{l}{2}$, put $n=(s+1)l+(l_1+l_2)(r+1)-2l_2+1$, then there exists a $q$-ary MDS Euclidean self-dual code of length $n+1$.     
\end{theorem}

\begin{proof}
Please see \ref{app A} for the proof.
\end{proof}
\section{MDS Euclidean self-dual codes from multipliacative subgroups}
In this section, we will always assume $q=r^2$, where $r$ is an odd prime power with $r \equiv 3\pmod 4$. Let $\theta$ be a primitive element of $\mathbb{F}_q$, i.e., $\mathbb{F}_q^{*}=\langle\theta\rangle$. Under this assumption, we again provide some constructions of MDS Euclidean self-dual codes based on multiplicative subgroups of $\mathbb{F}_q^{*}$ and their cosets.
\subsection{The special case of three sets}
In this subsection, let 
$$\alpha=\theta^{u}, \beta=\theta^{v}, \gamma=\theta^{\frac{v}{2}},$$
where $u$, $v$ are two distinct divisors of $q-1$ and $v$ is even. Let $\langle\alpha\rangle$ and $\langle\beta\rangle$ be two multiplicative subgroups of $\mathbb{F}_q^{*}$ generated by $\alpha$ and $\beta$, respectively. Suppose $0\leq s\leq \frac{u}{\gcd(u,v)}$,  $0\leq t\leq \frac{v}{\gcd(u,v)}$ and  $0\leq s^{\prime}\leq \frac{u}{\gcd(u,v)}$. Denote $A_i=\beta^i\langle\alpha\rangle$,                                    $B_j=\alpha^j\langle\beta\rangle$ and $C_k={\gamma}^{2k+1}\langle\alpha\rangle$. Let 
\begin{equation}
 A\triangleq \bigcup_{i=0}^{s-1}A_i,\ B \triangleq \bigcup_{j=0}^{t-1}B_j, \ C\triangleq \bigcup_{k=0}^{{s^{\prime}}-1}C_k.  \label{eq15}
\end{equation}
\begin{lemma}[\cite{FXF21}]\label{lem9}
 Keep the above notations. Let $A$ and $B$ be defined as Eq. (\ref{eq15}), then
 $$|A|=s\frac{q-1}{u}, |B|=t\frac{q-1}{v} ~\text{and}~ |A\cap B|=\frac{(q-1)\gcd(u,v)}{uv}st.$$
\end{lemma}
Based on the sets defined by Eq. (\ref{eq15}) and Lemma \ref{lem9}, we provide some constructions of MDS Euclidean self-dual codes as follows.
\begin{theorem}\label{th4}
Let $q=r^2$ and $r \equiv 3\pmod 4$. Suppose that $u$ and $v$ are factors of $q-1$. Let $0\leq s\leq \frac{u}{\gcd(u,v)}$, $0\leq s^{\prime}\leq \frac{u}{\gcd(u,v)}$ and $0 \leq t\leq \frac{v}{\gcd(u,v)}$. Put
$$n=(s+s^{\prime})\frac{q-1}{u}+t\frac{q-1}{v}-2\frac{(q-1)\gcd(u,v)}{uv}st.$$
Suppose the following conditions hold: 

(1) both $u$ and $v$ are even;

(2) $2u|(r+1)v$, $v|(r-1)u$ and $4\nmid v$;

(3) both $\frac{(r+1)v}{2u}$ and $s+s^{\prime}$ are odd.\\
Then there exists a $q$-ary MDS Euclidean self-dual code of length $n$.

\end{theorem}
\begin{proof}
Choose 
\begin{equation*}
S=\left(A\backslash \left(B\cup C\right)\right)\bigcup \left(B\backslash \left(A\cup C\right)\right)\bigcup \left(C\backslash \left(A\cup B\right)\right)\bigcup\left(A\cap B\cap C\right),    
\end{equation*}
where $A$, $B$ and $C$ are defined as in Eq. (\ref{eq15}).
Firstly, we determine the size of $S$. It follows from Lemma \ref{lem9} that $|A\cap B|=\frac{(q-1)\gcd(u,v)}{uv}st$. Similar to the proof for Lemma \ref{lem9}, $|C|=s^{\prime}\frac{q-1}{u}$ can be obtained. Since each element of $A$ is a square in $\mathbb{F}_q$ and each element of $C$ is a non-square in $\mathbb{F}_q$, we can deduce that $A\cap C=\emptyset$. Similarly, $B\cap C=\emptyset$. Hence,
$$|S|=|A|+|B|+|C|-2|A\cap B|=n.$$

Secondly, based on Theorem \ref{th1}, we will show that $S$ is an evaluation set satisfying the conditions in Lemma \ref{lem2} or \ref{lem3}. For the convenience of later calculations, we make some preparations in the following.
For any $0\leq i\leq s-1$, $0\leq j\leq t-1$ and $0\leq k\leq s^{\prime}-1$, 
 $$\pi_{A_i}(x)=\prod_{l=0}^{\frac{q-1}{u}-1}(x-\beta^i\alpha^l)=x^{\frac{q-1}{u}}-\beta^{\frac{q-1}{u}i}, \ \delta_{A_i}(x)=\frac{q-1}{u}x^{\frac{q-1}{u}-1},$$
 $$\pi_{B_j}(x)=\prod_{l=0}^{\frac{q-1}{v}-1}(x-\alpha^j\beta^l)=x^{\frac{q-1}{v}}-\alpha^{\frac{q-1}{v}j}, \ \delta_{B_j}(x)=\frac{q-1}{v}x^{\frac{q-1}{v}-1},$$
$$\pi_{C_k}(x)=\prod_{l=0}^{\frac{q-1}{u}}(x-\gamma^{2k+1}\alpha^l)=x^{\frac{q-1}{u}}-\gamma^{\frac{q-1}{u}(2k+1)}, \ \delta_{C_k}(x)=\frac{q-1}{u}x^{\frac{q-1}{u}-1}.$$
~~\textbf{Step 1}: For any $\beta^i\alpha^j\in A \backslash (B \cup C)$, we will calculate $\eta\left(\delta_A(\beta^i\alpha^j)\pi_B(\beta^i\alpha^j)\pi_C(\beta^i\alpha^j)\right)$, where $0\leq i \leq s-1$, $0\leq j\leq \frac{q-1}{u}-1$. By Lemma \ref{lem1},
\begin{eqnarray*}
     \delta_{A}(\beta^i\alpha^j)&=& \delta_{A_i}(\beta^i\alpha^j)\prod_{l=0, l\neq i}^{s-1}\pi_{A_l}(\beta^i\alpha^j)\nonumber \\
     &=&\frac{q-1}{u}\left(\beta^i\alpha^j\right)^{\frac{q-1}{u}-1}\prod_{l=0, l\neq i}^{s-1}\left(\theta^{\frac{q-1}{u}iv}-\theta^{\frac{q-1}{u}lv}\right).
\end{eqnarray*}
Denote $\Omega=\prod\limits_{l=0, l\neq i}^{s-1}\left(\theta^{\frac{q-1}{u}iv}-\theta^{\frac{q-1}{u}lv}\right)$. Since $u\mid (r+1)v$,
 $\left(\theta^{\frac{q-1}{u}iv}\right)^{r+1}$ $=\left(\theta^{\frac{(r+1)v}{u}i}\right)^{q-1}=1$.
 Hence,
\begin{eqnarray*}
  \Omega^r&=&\prod_{l=0, l\neq i}^{s-1}\left(\left(\theta^{\frac{q-1}{u}iv}\right)^r-\left(\theta^{\frac{q-1}{u}lv}\right)^r\right)\\
  &=&\prod_{l=0, l\neq i}^{s-1}\left(\theta^{-\frac{q-1}{u}iv}-\theta^{-\frac{q-1}{u}lv}\right)\\
  &=&\prod_{l=0, l\neq i}^{s-1}\theta^{-\frac{q-1}{u}v(i+l)}\left(\theta^{\frac{q-1}{u}lv}-\theta^{\frac{q-1}{u}iv}\right)\\
  &=&(-1)^{s-1}\theta^{-\frac{q-1}{u}v\left((s-2)i+\frac{s(s-1)}{2}\right)}\Omega,
\end{eqnarray*}
which implies that there exists an integer $k$ such that
\begin{equation*}
\Omega=\theta^{\frac{r+1}{2}(s-1)-\frac{v(r+1)}{u}\left((s-2)i+\frac{s(s-1)}{2}\right)+k(r+1)}.    
\end{equation*}
Thus,  
\begin{equation}
 \delta_{A}(\beta^i\alpha^j)=\frac{q-1}{u}\theta^{e},\label{eq16}
\end{equation}
where $e=(\frac{q-1}{u}-1)vi-uj+\frac{r+1}{2}(s-1)-\frac{v(r+1)}{u}\left((s-2)i+\frac{s(s-1)}{2}\right)+k(r+1)$.
By the conditions (1) and (2), $u$, $v$ and $\frac{v(r+1)}{u}$ are even.  Since $r \equiv 3\pmod 4$, $\frac{r+1}{2}$ is even. Thus, from Eq. (\ref{eq16}), 
$\eta( \delta_{A}(\beta^i\alpha^j))=1$ can be obtained.
Based on Lemma \ref{lem1},
 $$\pi_B(\beta^i\alpha^j)=\prod_{l=0}^{t-1}\pi_{B_l}(\beta^i\alpha^j)=\prod_{l=0}^{t-1}\left(\theta^{\frac{q-1}{v}uj}-\theta^{\frac{q-1}{v}ul}\right).$$
 Since $v\mid u(r-1)$, we have
 $\left(\theta^{\frac{q-1}{v}uj}\right)^{r-1}=\left(\theta^{\frac{u(r-1)}{v}j}\right)^{q-1}=1$ and $\left(\theta^{\frac{q-1}{v}ul}\right)^{r-1}=1$,
which implies that $\pi_B(\beta^i\alpha^j)\in \mathbb{F}_r^{*}$, i.e., $\eta\left(\pi_B(\beta^i\alpha^j)\right)=1$.
According to Lemma \ref{lem1},
\begin{eqnarray*}
  \pi_{C}(\beta^i\alpha^j)&=&\prod_{l=0}^{s^{\prime}-1}\pi_{C_l}(\beta^i\alpha^j)
  =\prod_{l=0}^{s^{\prime}-1}\left(\theta^{\frac{q-1}{u}vi}-\theta^{\frac{q-1}{u}\cdot\frac{v}{2}(2l+1)}\right).
\end{eqnarray*}
Since $2u\mid(r+1)v$, we have
 $$(\theta^{\frac{q-1}{u}vi})^{r+1}=(\theta^{\frac{v(r+1)}{u}i})^{q-1}=1,
 \ (\theta^{\frac{q-1}{u}\cdot\frac{v}{2}(2l+1)})^{r+1}=(\theta^{\frac{v(r+1)}{2u}(2l+1)})^{q-1}=1.$$
Similarly, there exists an integer $k$ such that
 $$\pi_{C}(\beta^i\alpha^j)=\theta^{\frac{r+1}{2}s^{\prime}-\frac{(r+1)v}{2u}\left(2is^{\prime}+s^{\prime}+(s^{\prime}-1)s^{\prime}\right)+k(r+1)}.$$
From $r \equiv 3\pmod 4$, we can deduce that $\eta\left(\pi_{C}(\beta^i\alpha^j)\right)=\eta\left(\theta^{\frac{(r+1)v}{2u}s^{\prime}}\right)$. Hence,
\begin{equation}
  \eta\left(\delta_{A}(\beta^i\alpha^j)\pi_B(\beta^i\alpha^j)\pi_{C}(\beta^i\alpha^j)\right)=\eta\left(\theta^{\frac{(r+1)v}{2u}s^{\prime}}\right).  \label{eq17} 
\end{equation}
~~\textbf{Step 2}: For any $\alpha^i\beta^j\in B \backslash (A \cup C)$, we will calculate $\eta\left(\pi_A(\alpha^i\beta^j)\delta_B(\alpha^i\beta^j)\pi_C(\alpha^i\beta^j)\right)$, where $0\leq i \leq t-1$, $0\leq j\leq \frac{q-1}{v}-1$. By Lemma \ref{lem1},
\begin{eqnarray*}
 \pi_A(\alpha^i\beta^{j})=\prod_{l=0}^{s-1}\pi_{A_l}(\alpha^i\beta^{j})
 =\prod_{l=0}^{s-1}(\theta^{\frac{q-1}{u}vj}-\theta^{\frac{q-1}{u}vl}).
\end{eqnarray*}
Similarly, since $u\mid (r+1)v$, there exists an integer $k$ such that
$$\pi_A(\alpha^i\beta^{j})=\theta^{\frac{r+1}{2}s-\frac{v(r+1)}{u}\left(sj+\frac{s(s-1)}{2}\right)+k(r+1)}.$$
From $r \equiv 3\pmod 4$ and $2u\mid(r+1)v$,  we have $\eta(\pi_A(\alpha^i\beta^{j})=1$.
By Lemma \ref{lem1} again, 
\begin{eqnarray*}
 \delta_B(\alpha^i\beta^j)&=&\pi_{B_i}^{\prime}(\alpha^i\beta^j)\prod_{l=0, l\neq i}^{t-1}\pi_{B_l}(\alpha^i\beta^j)\\
 &=&\frac{q-1}{v}\left(\alpha^i\beta^j\right)^{\frac{q-1}{v}-1}\prod_{l=0, l\neq i}^{t-1}\left(\theta^{\frac{q-1}{v}iu}-\theta^{\frac{q-1}{v}lu}\right).
\end{eqnarray*}
Denote $\Omega=\prod\limits_{l=0, l\neq i}^{t-1}\left(\theta^{\frac{q-1}{v}iu}-\theta^{\frac{q-1}{v}lu}\right)$. Since $v\mid (r-1)u$, $\Omega\in \mathbb{F}_r^{*}$, which implies that
 $\eta\left(\delta_B(\alpha^i\beta^j)\right)=1.$
According to  Lemma \ref{lem1},
 $$\pi_{C}(\alpha^i\beta^j)=\prod_{l=0}^{s^{\prime}-1}\left(\theta^{\frac{q-1}{u}vj}-\theta^{\frac{q-1}{u}\cdot\frac{v}{2}(2l+1)}\right).$$
Since $2u\mid(r+1)v$, there exists an integer $k$ such that
 $$\pi_{C}(\alpha^i\beta^j)=\theta^{\frac{r+1}{2}s^{\prime}-\frac{(r+1)v}{2u}\left(2js^{\prime}+s^{\prime}+(s^{\prime}-1)s^{\prime}\right)+k(r+1)}.$$
 Again by $r \equiv 3\pmod 4$, we can show that
\begin{equation}
  \eta\left(\pi_{A}(\alpha^i\beta^j)\delta_B(\alpha^i\beta^j)\pi_{C}(\alpha^i\beta^j)\right)=\eta\left(\theta^{\frac{(r+1)v}{2u}s^{\prime}}\right).  \label{eq18} 
\end{equation}
~~\textbf{Step 3}: For any $\gamma^{2i+1}\alpha^j\in C \backslash (A \cup B)$, we will calculate $\eta\left(\pi_A(\gamma^{2i+1}\alpha^j)\pi_B(\gamma^{2i+1}\alpha^j)\delta_C(\gamma^{2i+1}\alpha^j)\right)$, where $0\leq i \leq s^{\prime}-1$, $0\leq j\leq \frac{q-1}{u}-1$. By Lemma \ref{lem1},
\begin{eqnarray*}
\pi_A(\gamma^{2i+1}\alpha^j)
&=&\prod_{l=0}^{s-1}(\theta^{\frac{v}{2}\cdot\frac{q-1}{u}(2i+1)}-\theta^{\frac{q-1}{u}lv}).
\end{eqnarray*}
Since $2u\mid(r+1)v$, there exists an integer $k$ such that
$$\pi_A(\gamma^{2i+1}\alpha^j)=\theta^{\frac{r+1}{2}s-\frac{v(r+1)}{2u}\left((2i+1)s+s(s-1)\right)+k(r+1)}.$$
Similarly, we have $\eta\left(\pi_A(\gamma^{2i+1}\alpha^j)\right)=\eta\left(\theta^{\frac{(r+1)v}{2u}s}\right).$
Since $v\mid u(r-1)$, we can deduce that
\begin{eqnarray*}
 \pi_B(\gamma^{2i+1}\alpha^j)&=&\prod_{l=0}^{t-1}\pi_{B_l}(\gamma^{2i+1}\alpha^j)\\
 &=&\prod_{l=0}^{t-1}\left((\theta^{\frac{r-1}{2}(2i+1)+\frac{(r-1)u}{v}j})^{r+1}-(\theta^{\frac{(r-1)u}{v}l})^{r+1}\right)\in \mathbb{F}_r^{*}.
\end{eqnarray*} 
Hence, we can get $\eta\left(\pi_B(\gamma^{2i+1}\alpha^j)\right)=1.$
By Lemma \ref{lem1} again,
\begin{eqnarray*}
     \delta_{C}(\gamma^{2i+1}\alpha^j)&=& \pi_{C_i}^{\prime}(\gamma^{2i+1}\alpha^j)\prod_{l=0, l\neq i}^{s^{\prime}-1}\pi_{C_l}(\gamma^{2i+1}\alpha^j)\\
     &=&\frac{q-1}{u}\left(\gamma^{2i+1}\alpha^j\right)^{\frac{q-1}{u}-1}\prod_{l=0, l\neq i}^{s^{\prime}-1}\left(\theta^{\frac{v}{2}\cdot\frac{q-1}{u}(2i+1)}-\theta^{\frac{v}{2}\cdot\frac{q-1}{u}(2l+1)}\right).
\end{eqnarray*}
Denote $\Omega=\prod\limits_{l=0,l\neq i}^{s^{\prime}-1}\left(\theta^{\frac{v}{2}\cdot\frac{q-1}{u}(2i+1)}-\theta^{\frac{v}{2}\cdot\frac{q-1}{u}(2l+1)}\right)$. Since $2u\mid(r+1)v$, there exists an integer $k$ such that 
 $$\Omega=\theta^{\frac{r+1}{2}(s^{\prime}-1)-\frac{v(r+1)}{2u}\left(2is^{\prime}+(s^{\prime}-1)s^{\prime}+2s^{\prime}-4i-2\right)+k(r+1)},$$
we can deduce that $\eta(\Omega)=1$. By the condition (2), $\frac{q-1}{u}-1$ is odd, which implies that
$\eta\left(\delta_{C}(\gamma^{2i+1}\alpha^j)\right)=-1$.
Furthermore, we have
\begin{equation}
 \eta\left(\pi_A(\gamma^{2i+1}\alpha^j)\pi_B(\gamma^{2i+1}\alpha^j)\delta_C(\gamma^{2i+1}\alpha^j)\right)=-\eta\left(\theta^{\frac{(r+1)v}{2u}s}\right).\label{eq19} 
\end{equation}
When both $\frac{(r+1)v}{2u}$ and $s+s^{\prime}$ are odd, we have $\eta\left(\theta^{\frac{(r+1)v}{2u}s^{\prime}}\right)=-\eta\left(\theta^{\frac{(r+1)v}{2u}s}\right)$. According to Theorem \ref{th1} and Eqs. (\ref{eq17}), (\ref{eq18}) and (\ref{eq19}), $\eta(\delta_S(e))=\eta\left(\theta^{\frac{(r+1)v}{2u}s^{\prime}}\right)$ for any $e\in S$. By Lemma \ref{lem2}, there exists a $q$-ary MDS Euclidean self-dual code of length $n$.
\end{proof}
 When we slightly modify the conditions in Theorem \ref{th4}  and add the zero element into the set $A$, MDS Euclidean self-dual code of length $n+2$ are obtained.
\begin{corollary}\label{cor1}
Let $q=r^2$ and $r \equiv 3 \pmod 4$. Suppose $u$ and $v$ are factors of $q-1$. Let $0\leq s\leq \frac{u}{\gcd(u,v)}$, $0\leq s^{\prime}\leq \frac{u}{\gcd(u,v)}$ and $0 \leq t\leq \frac{v}{\gcd(u,v)}$. Put
$$n=(s+s^{\prime})\frac{q-1}{u}+t\frac{q-1}{v}-2\frac{(q-1)\gcd(u,v)}{uv}st.$$
Suppose the following conditions hold: 

(1) both $u$ and $v$ are even;

(2) $2u|(r+1)v$, $v|(r-1)u$ and $4\nmid v$;

(3) both $s$ and $s^{\prime}$ are even or $\frac{(r+1)v}{2u}$ is even.\\
Then there exists a $q$-ary MDS Euclidean self-dual code of length $n+2$.
\end{corollary}
\begin{proof}
Please see \ref{app B} for the proof.
\end{proof}
By modifying the constraints on the two factors $u$ and $v$ proposed in Corollary \ref{cor1},  we can also  provide the following construction of MDS Euclidean self-dual codes with length $n+1$ or $n+2$.
\begin{theorem}\label{th5}
Let $q=r^2$ and $r\equiv 3\pmod 4$. Suppose that $u$ and $v$ are factors of $q-1$. Let $0\leq s, s^{\prime}\leq \frac{u}{\gcd(u,v)}$ and  $0\leq t\leq \frac{v}{\gcd(u,v)}$. Put 
$$n=(s+s^{\prime})\frac{q-1}{u}+t\frac{q-1}{v}-\frac{2(q-1)\gcd(u,v)}{uv}st.$$
Suppose the following conditions hold:

(1) $2^l|u$,  $v \equiv2^l\pmod {2^{l+1}}$ for some $l\geq 2$ ;

(2) $2u|v(r+1)$, $v|u(r-1)$;

(3) both $s$ and $s^{\prime}$ are even or $\frac{v(r+1)}{2u}$ is even.\\
When $n$ is odd, there exists a $q$-ary MDS Euclidean self-dual code of length $n+1$. When $n$ is even, there exists a $q$-ary MDS Euclidean self-dual code of length $n+2$.
\end{theorem}
\begin{proof} 
Choose $$S=(A^{\prime}\backslash (B\cup C))\bigcup (B\backslash (A^{\prime}\cup C))\bigcup (C\backslash (A^{\prime}\cup B))\bigcup(A^{\prime}\cap B\cap C),$$ 
where $A$, $B$ and $C$ are defined as in Eq. (\ref{eq18}) and $A^{\prime}=A\cup\{0\}$.
Firstly, we determine the size of $S$. By Lemma 6 of \cite{FXF21}, $|A\cap B|=\frac{(q-1)\gcd(u,v)}{uv}st$. Next we will prove $A\cap C=\emptyset$. Otherwise, let $e \in A\cap C$, then there exist some integers $0\leq i_1\leq s-1$, $1\leq j_1\leq \frac{q-1}{u}$, $0\leq i_2\leq s^{\prime}-1$ and $1\leq j_2\leq \frac{q-1}{u}$ such that
\begin{equation}
 e=\beta^{i_1}\alpha^{j_1}=\gamma^{2i_2+1}\alpha^{j_2}. \label{eq20}  
\end{equation}
Note that Eq. (\ref{eq20}) is equivalent to
$$\theta^{u(j_1-j_2)+v(i_1-i_2-\frac{1}{2})}=1.$$    
Thus, $u|\frac{v}{2}(2i_1-2i_2-1)$, which implies that $\frac{u}{\gcd(u,\frac{v}{2})}|(2i_1-2i_2-1)$. Since $2^l\nmid \frac{v}{2}$ and $2^l\mid u$, we have $2\mid \frac{u}{\gcd(u,\frac{v}{2})}$, implying $2\mid (2i_1-2i_2-1)$, which is a contradiction. Hence, $A\cap C=\emptyset$. Similarly, $B\cap C=\emptyset$.
Hence,
$$|S|=|A|+|B|+|C|-2|A\cap B|+1=n+1.$$

Secondly, based on Theorem \ref{th1}, we will show that $S$ is an evaluation set satisfying the conditions of Lemma \ref{lem2} or \ref{lem3}.

Since $2^l|u$ and $v \equiv 2^l \pmod {2^{l+1}}$ for some $l\geq 2$, both $u$ and $v$ are even. For any $\beta^i\alpha^j\in A \backslash (B \cup C)$, where $0\leq i \leq s-1$, $0\leq j\leq \frac{q-1}{u}-1$, it follows from Eq. (\ref{eq17}) and the condition (2) that 
\begin{equation}
  \eta\left(\delta_{A}(\beta^i\alpha^j)\pi_B(\beta^i\alpha^j)\pi_{C}(\beta^i\alpha^j)\right)=\eta\left(\theta^{\frac{(r+1)v}{2u}s^{\prime}}\right).\label{eq21} 
\end{equation}
Similarly, for any $\alpha^i\beta^j\in B \backslash (A \cup C)$, where $0\leq i \leq t-1$, $0\leq j\leq \frac{q-1}{v}-1$,
\begin{equation*}
  \eta\left(\pi_{A}(\alpha^i\beta^j)\delta_B(\alpha^i\beta^j)\pi_{C}(\alpha^i\beta^j)\right)=\eta\left(\theta^{\frac{(r+1)v}{2u}s^{\prime}}\right).   
\end{equation*}
By $v \equiv 2^l \pmod {2^{l+1}}$ for some $l\geq 2$,  we can deduce that $4|v$. Thus, $\gamma$ is a square element in $\mathbb{F}_q$. According to the proof of Theorem \ref{th4} and the condition (2), for any $\gamma^{2i+1}\alpha^j\in C \backslash (A \cup B)$, where $0\leq i \leq s^{\prime}-1$, $0\leq j\leq \frac{q-1}{u}-1$,
\begin{equation*}
 \eta\left(\pi_A(\gamma^{2i+1}\alpha^j)\pi_B(\gamma^{2i+1}\alpha^j)\delta_C(\gamma^{2i+1}\alpha^j)\right)=\eta\left(\theta^{\frac{(r+1)v}{2u}s}\right). 
\end{equation*}

 For any $a\in A^{\prime} \backslash (B \cup C)$, if $a=0$, then by the proof of Corollary \ref{cor1}, 
\begin{equation*}
 \eta\left(\delta_{A^{\prime}}(0)\pi_B(0)\pi_C(0)\right)=1. 
\end{equation*}
Otherwise, for any $\beta^i\alpha^j\in A \backslash (B \cup C)$, where $0\leq i \leq s-1$, $0\leq j\leq \frac{q-1}{u}-1$,
\begin{equation*}
\delta_{A^{\prime}}(\beta^i\alpha^j)\pi_B(\beta^i\alpha^j)\pi_{C}(\beta^i\alpha^j)=\beta^{i}\alpha^{j} \cdot \delta_{A}(\beta^i\alpha^j)\pi_B(\beta^i\alpha^j)\pi_{C}(\beta^i\alpha^j).   
\end{equation*}
According to the condition (1) and Eq. (\ref{eq21}),
\begin{equation*}
  \eta\left(\delta_{A^{\prime}}(\beta^i\alpha^j)\pi_B(\beta^i\alpha^j)\pi_{C}(\beta^i\alpha^j)\right)=\eta\left(\theta^{\frac{(r+1)v}{2u}s^{\prime}}\right).   
\end{equation*}
Similarly, for any $\alpha^i\beta^j\in B \backslash (A \cup C)$, where $0\leq i \leq t-1$, $0\leq j\leq \frac{q-1}{v}-1$,
\begin{equation*}
  \eta\left(\pi_{A^{\prime}}(\alpha^i\beta^j)\delta_B(\alpha^i\beta^j)\pi_{C}(\alpha^i\beta^j)\right)=\eta\left(\theta^{\frac{(r+1)v}{2u}s^{\prime}}\right).   
\end{equation*}
Finally, for any $\gamma^{2i+1}\alpha^j\in C \backslash (A \cup B)$, where $0\leq i \leq s^{\prime}-1$, $0\leq j\leq \frac{q-1}{u}-1$,
\begin{equation*}
 \pi_{A^{\prime}}(\gamma^{2i+1}\alpha^j)\pi_B(\gamma^{2i+1}\alpha^j)\delta_C(\gamma^{2i+1}\alpha^j)=\gamma^{2i+1}\alpha^j \cdot \pi_{A}(\gamma^{2i+1}\alpha^j)\pi_B(\gamma^{2i+1}\alpha^j)\delta_C(\gamma^{2i+1}\alpha^j). 
\end{equation*}
Since  $\gamma$ is a square element in $\mathbb{F}_q$, we obtain
\begin{equation*}
 \eta\left(\pi_{A^{\prime}}(\gamma^{2i+1}\alpha^j)\pi_B(\gamma^{2i+1}\alpha^j)\delta_C(\gamma^{2i+1}\alpha^j)\right)=\eta\left(\theta^{\frac{(r+1)v}{2u}s}\right). 
\end{equation*}

When $\frac{(r+1)v}{2u}$ is even or both $s$ and $s^{\prime}$ are even, we can deduce that $\eta\left(\theta^{\frac{(r+1)v}{2u}s^{\prime}}\right)=\eta\left(\theta^{\frac{(r+1)v}{2u}s}\right)=1$. According to Theorem \ref{th1}, we have $\eta(\delta_S(e))=\eta(-\delta_S(e))=1$ for any $e\in S$. When $n$ is odd, by Lemma \ref{lem2}, there exists a $q$-ary MDS Euclidean self-dual code of length $n+1$. When $n$ is even, by Lemma \ref{lem3}, there exists a $q$-ary MDS Euclidean self-dual code of length $n+2.$
\end{proof}
\subsection{The general case of three sets}
In this subsection, we let 
$$\alpha=\theta^{u}, \beta=\theta^{v}, \gamma=\theta^{w},$$
where $u$, $v$ and $w$ are three divisors of $q-1$. Let $\langle\alpha\rangle$, $\langle\beta\rangle$ and $\langle\gamma\rangle$ be the three multiplicative subgroups of $\mathbb{F}_q^{*}$ generated by $\alpha$, $\beta$ and $\gamma$, respectively. Suppose $0\leq s\leq \frac{u}{\gcd(u,v)}$,  $0\leq t\leq \frac{v}{\gcd(v,w)}$ and  $0\leq f\leq \frac{w}{\gcd(w,u)}$. Denote $A_i=\beta^i\langle\alpha\rangle$,                                    $B_j=\gamma^j\langle\beta\rangle$ and $C_k={\alpha}^{k}\langle\gamma\rangle$. Let 
\begin{equation}
 A\triangleq \bigcup_{i=0}^{s-1}A_i,\ B \triangleq \bigcup_{j=0}^{t-1}B_j, \ C\triangleq \bigcup_{k=0}^{{f-1}} C_k.  \label{eq22}
\end{equation}

{In Lemma \ref{lem9}, the authors given the sizes of $A$, $B$ and $C$, then the following lemma gives the cardinality of $A\cap B$, $B\cap C$, $C\cap A$ and $A\cap B\cap C$.}
\begin{lemma}\label{lem10}
 Keep the above notations. Let $s^{\prime}=\left\lfloor\frac{(s-1)\gcd(u,v,w)}{\gcd(w,u)}\right\rfloor+1$, $t^{\prime}=\left\lfloor\frac{(t-1)\gcd(u,v,w)}{\gcd(u,v)}\right\rfloor+1$ and $f^{\prime}=\left\lfloor\frac{(f-1)\gcd(u,v,w)}{\gcd(v,w)}\right\rfloor+1$,
 then 
 $$|A\cap B|=\frac{(q-1)}{\operatorname{lcm}(u,v)}st^{\prime}, \ |B\cap C|=\frac{(q-1)}{\operatorname{lcm}(v,w)}tf^{\prime}, \ |C\cap A|=\frac{(q-1)}{\operatorname{lcm}(w,u)}fs^{\prime}$$
 and
 $$|A\cap B\cap C|=\frac{(q-1)}{\operatorname{lcm}(u,v,w) }s^{\prime}t^{\prime}f^{\prime}.$$
\end{lemma}
\begin{proof}
 Firstly, we consider the size of $A\cap B$. For any $0\leq i\leq s-1$, we think about what conditions $j$ should satisfy when $A_i\cap B_j\neq \emptyset$. Now let $e\in A_i\cap B_j$, then there exist some $0\leq i_1\leq \frac{q-1}{u}-1$ and $0\leq j_1\leq \frac{q-1}{v}-1$ such that
 $e=\beta^{i}\alpha^{i_1}=\gamma^{j}\beta^{j_1}$,
 which is equivalent to
 $$\theta^{v(i-j_1)+ui_1-wj}=1.$$
Thus, we have $\gcd(u,v)\mid wj$, which implies that $\frac{\gcd(u,v)}{\gcd(u,v,w)}\mid j$. 

For $0\leq i\leq s-1$ and $0\leq j\leq t-1$ with $\frac{\gcd(u,v)}{\gcd(u,v,w)}\mid j$, we will calculate the size of $A_i\cap B_j$ for some fixed $i$, $j$, $k$. Let $x\equiv iv+i_1u\equiv jw+j_1v \pmod {q-1}$, then the size of $A_i\cap B_j$ is equal to the number of solutions of Eq. (\ref{eq23}) in the sense of modulo $q-1$.
\begin{equation}
          \begin{cases}
             x  \equiv iv \ ( {\rm mod} \ u),\\
             x \equiv jw   \ ({\rm mod} \ v).
          \end{cases}\label{eq23}
\end{equation}
Because $\frac{\gcd(u,v)}{\gcd(u,v,w)}\mid j$, then we have $\gcd(u,v)\mid jw$. In order to study the solution of Eq. (\ref{eq23}), we first consider the solution of the following system of equations.
\begin{equation}
          \begin{cases}
             \frac{x}{\gcd(u,v)}  \equiv \frac{iv}{\gcd(u,v)} \ ( {\rm mod} \ \frac{u}{\gcd(u,v)}),\\
             \frac{x}{\gcd(u,v)} \equiv \frac{jw}{\gcd(u,v)}   \ ({\rm mod} \ \frac{v}{\gcd(u,v)}).
          \end{cases}
\end{equation}
Denote $y_1\equiv (\frac{v}{\gcd(u,v)})^{-1} \pmod{ \frac{u}{\gcd(u,v)}}$ and $y_2\equiv (\frac{u}{\gcd(u,v)})^{-1} \pmod{\frac{v}{\gcd(u,v)}}$.
Since \\$\gcd(\frac{u}{\gcd(u,v)},\frac{v}{\gcd(u,v)})=1$, by the {Chinese remainder theorem}, Eq. (\ref{eq24}) has a unique solution
$$\frac{x}{\gcd(u,v)}\equiv y_1\cdot \frac{v}{\gcd(u,v)}\frac{iv}{\gcd(u,v)}+y_2 \cdot\frac{u}{\gcd(u,v)}\frac{jw}{\gcd(u,v)}\pmod{\frac{\operatorname{lcm}(u,v)}{\gcd(u,v)}}.$$
Thus, Eq. (\ref{eq23}) has a unique solution as follows
\begin{equation}
 x\equiv y_1 \cdot\frac{v}{\gcd(u,v)}{iv}+y_2 \cdot\frac{u}{\gcd(u,v)}{jw}\pmod{\operatorname{lcm}(u,v)},  \label{eq25} 
\end{equation}
which implies that Eq. (\ref{eq23}) has $\frac{(q-1)}{\operatorname{lcm}(u,v)}$ solutions in the sense of modulo $q-1$, i.e., $|A_i\cap B_j|=\frac{q-1}{\operatorname{lcm}(u,v)}$. From $0\leq j\leq  t-1$, the number of $j$ satisfying $\frac{\gcd(u,v)}{\gcd(u,v,w)}\mid j$ is $\left\lfloor(t-1)/\frac{\gcd(u,v)}{\gcd(u,v,w)}\right\rfloor+1$, which shows that $|A\cap B|=\frac{(q-1)}{\operatorname{lcm}(u,v)}\left(\left\lfloor\frac{(t-1)\gcd(u,v,w)}{\gcd(u,v)}\right\rfloor+1\right)s$. Similarly, we can give the size of $B\cap C$ and $C\cap A$.

Next, we consider the size of $A\cap B\cap C$. For any $0\leq i\leq s-1$, $0\leq j\leq t-1$ and $0\leq k\leq f-1$, if $A_i\cap B_j\cap C_k\neq \emptyset$, then we must have $A_i\cap B_j\neq \emptyset$, $B_j\cap C_k\neq \emptyset$ and $C_k\cap A_i\neq \emptyset$, which implies that $\frac{\gcd(u,w)}{\gcd(u,v,w)}\mid i$, $\frac{\gcd(u,v)}{\gcd(u,v,w)}\mid j$ and $\frac{\gcd(v,w)}{\gcd(u,v,w)}\mid k$.

For $i$, $j$ and $k$ that satisfy the above conditions, we will calculate the size of $A_i\cap B_j\cap C_k$ for some fixed $i$, $j$, $k$. Let $x\equiv iv+i_1u\equiv jw+j_1v\equiv kw+k_1u \pmod {q-1}$, where $0\leq i_1\leq \frac{q-1}{u}-1$, $0\leq j_1\leq \frac{q-1}{v}-1$ and $0\leq k_1\leq \frac{q-1}{w}-1$. It is clear that the size of $A_i\cap B_j\cap C_k$ is equal to the number of solutions of Eq. (\ref{eq26}) in the sense of modulo $q-1$.
\begin{equation}
    \begin{cases}
     x  \equiv iv \ ( {\rm mod} \ u),\\
              x \equiv jw   \ ({\rm mod} \ v),\\
              x \equiv ku   \ ({\rm mod} \ w).   
    \end{cases}\label{eq26}
\end{equation}
According to Eq. (\ref{eq25}), we can know that Eq. (\ref{eq26}) is equivalent to the following  system of equations,
\begin{equation}
    \begin{cases}
      x  \equiv  y_1 \cdot\frac{v}{\gcd(u,v)}{iv}+y_2 \cdot\frac{u}{\gcd(u,v)}{jw}\ ({\rm mod}\ \operatorname{lcm}(u,v)),\\
              x \equiv ku  \ ({\rm mod} \ w).   
    \end{cases}\label{eq27}
\end{equation}
In order to study the solution of Eq. (\ref{eq27}), we will prove that  $$\gcd(\operatorname{lcm}(u,v),w)\mid \frac{u}{\gcd(u,v)}{jw}, ~\gcd(\operatorname{lcm}(u,v),w)\mid \frac{v}{\gcd(u,v)}{iv}~ \text{and} ~ \gcd(\operatorname{lcm}(u,v),w)\mid ku.$$ 
Note that $\gcd(\operatorname{lcm}(u,v),w)\mid \frac{u}{\gcd(u,v)}{jw}$ is clear.
Since $\gcd(u,\frac{v}{\gcd(u,v)})=1$, $\gcd(\operatorname{lcm}(u,v),w)=\gcd(\frac{uv} {\gcd(u,v)},w)=\gcd(u,w)\gcd(\frac{v}{\gcd(u,v)},w)$ can be obtained. By the condition $\gcd(u,w)\mid iv$, we can deduce that $\gcd(\operatorname{lcm}(u,v),w)\mid \frac{v}{\gcd(u,v)}{iv}$. 
Since 
$$\gcd(\operatorname{lcm}(u,v),w)=\operatorname{lcm}(\gcd(u,w),\gcd(v,w))=\frac{\gcd(u,w)\gcd(v,w)}{\gcd(u,v,w)},$$
and
$$\gcd(\operatorname{lcm}(u,v),w,u)=\gcd(u,w),$$
we obtain $\frac{\gcd(\operatorname{lcm}(u,v),w)}{\gcd(\operatorname{lcm}(u,v),w,u)}=\frac{\gcd(v,w)}{\gcd(u,v,w)}$. By the condition $\frac{\gcd(v,w)}{\gcd(u,v,w)}\mid k$, we can deduce that $\gcd(\operatorname{lcm}(u,v),w)\mid ku$. According to the {Chinese remainder theorem}, Eq. (\ref{eq27}) has a unique solution in the sense of modulo $\operatorname{lcm}(u,v,w)$, 
implying that Eq. (\ref{eq27}) has $\frac{q-1}{\operatorname{lcm}(u,v,w)}$ solutions in the sense of modulo $q-1$, i.e.,
$$|A_i\cap B_j\cap C_k|=\frac{q-1}{\operatorname{lcm}(u,v,w)}.$$ 
Similar to the calculation of $|A\cap B|$, we have $|A\cap B\cap C|=\frac{(q-1)}{\operatorname{lcm}(u,v,w) }s^{\prime}t^{\prime}f^{\prime}$. 
\end{proof} 
Based on the three sets $A$, $B$ and $C$ defined by Eq. (\ref{eq22}), we provide a generic construction of MDS Euclidean self-dual codes as follows.
\begin{theorem}\label{th6}
 Let $q=r^2$ and $r \equiv 3 \pmod 4$. Suppose that $u$, $v$ and $w$ are three factors of $q-1$. Let $0\leq s\leq \frac{u}{\gcd(u,v)}$, $0 \leq t\leq \frac{v}{\gcd(v,w)}$ and $0\leq f\leq \frac{w}{\gcd(w,u)}$. Put
$$n=s\frac{q-1}{u}+t\frac{q-1}{v}+f\frac{q-1}{w}-\frac{2(q-1)st^{\prime}}{\operatorname{lcm}(u,v)}-\frac{2(q-1)tf^{\prime}}{\operatorname{lcm}(v,w)}-\frac{2(q-1)fs^{\prime}}{\operatorname{lcm}(w,u)}+\frac{4(q-1)s^{\prime}t^{\prime}f^{\prime}}{\operatorname{lcm}(u,v,w) }.$$
Suppose the following conditions hold: 

(1) $n$, $u$, $v$ and $w$ are even; 

(2) $u\mid (r+1)v $, $u\mid (r+1)w$, $v\mid (r-1)u $, $v\mid (r-1)w$, $w\mid (r-1)u $ and  $w\mid (r-1)v$;

(3) both $\frac{(r+1)v}{u}s$ and $\frac{w(r+1)}{u}s$ are even. \\
Then there exists a $q$-ary MDS Euclidean self-dual code of length $n$.   
\end{theorem}
\begin{proof}
Choose 
\begin{equation*}
S=\left(A\backslash \left(B\cup C\right)\right)\bigcup \left(B\backslash \left(A\cup C\right)\right)\bigcup \left(C\backslash \left(A\cup B\right)\right)\bigcup\left(A\cap B\cap C\right),    
\end{equation*}
where $A$, $B$ and $C$ are defined as in Eq. (\ref{eq22}).
By Lemma \ref{lem10}, 
$$|S|=|A|+|B|+|C|-2|A\cap B|-2|B\cap C|-2|C\cap A|+4|A\cap B\cap C|=n$$
~~Firstly, for any $\beta^{i_1}\alpha^{j_1}\in A \backslash (B \cup C)$, we will calculate $\eta(\delta_{A}(\beta^{i_1}\alpha^{j_1})\pi_B(\beta^{i_1}\alpha^{j_1})\pi_C(\beta^{i_1}\alpha^{j_1}))$, where $0\leq i_1 \leq s-1$, $0\leq j_1\leq \frac{q-1}{u}-1$. Based on Eq. (\ref{eq16}), 
\begin{equation}
 \delta_{A}(\beta^{i_1}\alpha^{j_1})=\frac{q-1}{u}\theta^{e_1},\label{eq28}
\end{equation}
where $e_1=(\frac{q-1}{u}-1)vi_1-uj_1+\frac{r+1}{2}(s-1)-\frac{v(r+1)}{u}\left((s-2)i_1+\frac{s(s-1)}{2}\right)+k(r+1)$. By the conditions $v\mid (r-1)u$ and $v\mid (r-1)w$, 
\begin{equation*}
\pi_B(\beta^{i_1}\alpha^{j_1})=\prod\limits_{l=0}^{t-1}\pi_{B_l}(\beta^{i_1}\alpha^{j_1})=\prod\limits_{l=0}^{t-1}   (\theta^{\frac{j_1u(r-1)}{v}(r+1)}-\theta^{\frac{lw(r-1)}{v}(r+1)}) \in \mathbb{F}_r^{*}.
\end{equation*}
Similarly, we can deduce that $\pi_C(\beta^{i_1}\alpha^{j_1})\in \mathbb{F}_r^{*}$. Thus, 
\begin{equation}
 \eta(\delta_{A}(\beta^{i_1}\alpha^{j_1})\pi_B(\beta^{i_1}\alpha^{j_1})\pi_C(\beta^{i_1}\alpha^{j_1}))=\eta(\theta^{e_1}).\label{eq29}   
\end{equation}
~~Secondly, for any $\gamma^{i_2}\beta^{j_2}\in B\backslash (A\cup C)$, we will calculate $\eta\left(\pi_A(\gamma^{i_2}\beta^{j_2})\delta_B(\gamma^{i_2}\beta^{j_2})\pi_{C}(\gamma^{i_2}\beta^{j_2})\right)$, where $0\leq i_2\leq t-1$, $0\leq j_2\leq \frac{q-1}{v}-1$. By Lemma \ref{lem1},
\begin{equation*}
\pi_A(\gamma^{i_2}\beta^{j_2})=\prod\limits_{l=0}^{s-1}\pi_{A_l}(\gamma^{i_2}\beta^{j_2})=\prod\limits_{l=0}^{s-1}(\theta^{\frac{q-1}{u}(wi_2+vj_2)}-\theta^{\frac{q-1}{u}vl}).   
\end{equation*}
Since $u\mid (r+1)v$ and $u\mid (r+1)w$, we can deduce that $\left(\theta^{\frac{q-1}{u}(wi_2+vj_2)}\right)^{r}=\theta^{-\frac{q-1}{u}(wi_2+vj_2)}$ and $\left(\theta^{\frac{q-1}{u}vl}\right)^{r}=\theta^{-\frac{q-1}{u}vl}$. Thus,
\begin{eqnarray*}
 \left(\pi_A(\gamma^{i_2}\beta^{j_2})\right)^r&=&\prod\limits_{l=0}^{s=1}\left(\theta^{-\frac{q-1}{u}(wi_2+vj_2)}-\theta^{-\frac{q-1}{u}vl}\right)\\
 &=&\prod\limits_{l=0}^{s-1}\theta^{-\frac{q-1}{u}(wi_2+vj_2+vl)}\left(\theta^{\frac{q-1}{u}vl}-\theta^{\frac{q-1}{u}(wi_2+vj_2)}\right)\\
 &=&(-1)^s\theta^{-\frac{q-1}{u}(swi_2+svj_2+\frac{s(s-1)}{2}v)}\pi_A(\gamma^{i_2}\beta^{j_2}),
\end{eqnarray*}
which implies that $\left(\pi_A(\gamma^{i_2}\beta^{j_2})\right)^{r-1}=(-1)^s\theta^{-\frac{q-1}{u}(swi_2+svj_2+\frac{s(s-1)}{2}v)}$. Note that $-1=\theta^{(r-1)\frac{r+1}{2}}$, $\frac{q-1}{u}w=(r-1)\frac{w(r+1)}{u}$ and $\frac{q-1}{u}v=(r-1)\frac{v(r+1)}{u}$, then there exists an integer $k^{\prime}$ such that
$$\pi_A(\gamma^{i_2}\beta^{j_2})=\theta^{\frac{r+1}{2}s-\frac{w(r+1)}{u}si_2-\frac{v(r+1)}{u}\left(sj_2+\frac{s(s-1)}{2}\right)+k^{\prime}(r+1)}.$$
According to Lemma \ref{lem1}, 
\begin{eqnarray*}
\delta_B(\gamma^{i_2}\beta^{j_2})&=&\delta_{B_{i_2}}(\gamma^{i_2}\beta^{j_2})\prod\limits_{l=0, l\neq i_2}^{t-1}\pi_{B_l}(\gamma^{i_2}\beta^{j_2})\\
&=&\frac{q-1}{v}(\gamma^{i_2}\beta^{j_2})^{\frac{q-1}{v}-1}\prod\limits_{l=0, l\neq i_2}^{t-1}\left((\gamma^{i_2}\beta^{j_2})^{\frac{q-1}{v}}-\gamma^{\frac{(q-1)l}{v}}\right)\\
&=&\frac{q-1}{v}\theta^{(\frac{q-1}{v}-1)i_2w-j_2v}\prod\limits_{ l=0, l\neq i_2}^{t-1}\left(\theta^{\frac{(q-1)w}{v}i_2}-\theta^{\frac{(q-1)w}{v}l}\right).
\end{eqnarray*}
Since $v\mid (r-1)w$, we have $\prod\limits_{l=0, l\neq i_2}^{t-1}\left(\theta^{\frac{(q-1)w}{v}i_2}-\theta^{\frac{(q-1)w}{v}l}\right)\in \mathbb{F}_r^{*}$, which implies that 
$$\eta\left(\delta_B(\gamma^{i_2}\beta^{j_2})\right)=\eta\left(\theta^{(\frac{q-1}{v}-1)i_2w-j_2v}\right).$$
Because $w\mid (r-1)u$ and $w\mid (r-1)v$, so we can deduce that
$$\pi_{C}(\gamma^{i_2}\beta^{j_2})=\prod\limits_{l=0}^{f-1}\pi_{C_l}=\prod\limits_{l=0}^{f-1}\left(\theta^{\frac{(q-1)v}{w}j_2}-\theta^{\frac{(q-1)u}{v}l}\right)\in \mathbb{F}_r^{*}.$$
Hence, we have 
\begin{equation}
 \eta\left(\pi_A(\gamma^{i_2}\beta^{j_2})\delta_B(\gamma^{i_2}\beta^{j_2})\pi_{C}(\gamma^{i_2}\beta^{j_2})\right)=\eta\left(\theta^{e_2}\right),   \label{eq30}
\end{equation}
where $e_2=\frac{r+1}{2}s-\frac{w(r+1)}{u}si_2-\frac{v(r+1)}{u}\left(sj_2+\frac{s(s-1)}{2}\right)+k^{\prime}(r+1)+(\frac{q-1}{v}-1)i_2w-j_2v$.

Next, for any $\alpha^{i_3}\gamma^{j_3}\in C\backslash (A\cup B)$, we will calculate $ \eta(\pi_{A}(\alpha^{i_3}\gamma^{j_3})\pi_{B}(\alpha^{i_3}\gamma^{j_3})\delta_C(\alpha^{i_3}\gamma^{j_3}))$, where $0\leq i_3\leq f-1$, $0\leq j_3\leq \frac{q-1}{w}-1$. By Lemma \ref{lem1},
\begin{equation*}
\pi_{A}(\alpha^{i_3}\gamma^{j_3})=\prod\limits_{l=0}^{s-1}\pi_{A_l}(\alpha^{i_3}\gamma^{j_3})= \prod\limits_{l=0}^{s-1}(\theta^{\frac{(q-1)w}{u}j_3}-\theta^{\frac{(q-1)v}{u}l}).   
\end{equation*}
By the conditions $u\mid (r+1)w$ and $u\mid (r+1)v$, we can know that there exists an integer $k^{\prime\prime}$ such that 
$$\pi_{A}(\alpha^{i_3}\gamma^{j_3})=\theta^{\frac{r+1}{2}s-\frac{(r+1)w}{u}sj_3-\frac{(r+1)v}{u}\frac{s(s-1)}{2}+k^{\prime\prime}(r+1)}.$$
Since $v\mid (r-1)u$ and $v\mid (r-1)w$, we can deduce that 
\begin{equation*}
\pi_{B}(\alpha^{i_3}\gamma^{j_3})=\prod\limits_{l=0}^{t-1}\pi_{B_l}(\alpha^{i_3}\gamma^{j_3})=\prod\limits_{l=0}^{t-1}(\theta^{\frac{(q-1)u}{v}i_3+\frac{(q-1)w}{v}j_3}-\theta^{\frac{(q-1)w}{v}l})\in \mathbb{F}_r^{*}.    
\end{equation*}
Based on Lemma \ref{lem1}, 
\begin{eqnarray*}
 \delta_C(\alpha^{i_3}\gamma^{j_3})&=&\delta_{C_{i_3}}(\alpha^{i_3}\gamma^{j_3})\prod\limits_{l=0, l\neq i_3}^{f-1}\pi_{C_l}(\alpha^{i_3}\gamma^{j_3})\\
 &=&\frac{q-1}{w}(\alpha^{i_3}\gamma^{j_3})^{\frac{q-1}{w}-1}\prod\limits_{l=0, l\neq i_3}^{f-1}\left((\alpha^{i_3}\gamma^{j_3})^{\frac{q-1}{w}}-\alpha^{\frac{q-1}{w}l}\right)\\
 &=&\frac{q-1}{w}\theta^{(\frac{q-1}{w}-1)i_3u-j_3w}\prod\limits_{l=0, l\neq i_3}^{f-1}\left(\theta^{\frac{(q-1)u}{w}i_3}-\theta^{\frac{(q-1)u}{w}l}\right).
\end{eqnarray*}
By the condition $w\mid (r-1)u$, we can know that $\prod\limits_{l=0, l\neq i_3}^{f-1}\left(\theta^{\frac{(q-1)u}{w}i_3}-\theta^{\frac{(q-1)u}{w}l}\right)\in \mathbb{F}_r^{*}$, which implies that $\eta(\delta_C(\alpha^{i_3}\gamma^{j_3}))=\eta(\theta^{(\frac{q-1}{w}-1)i_3u-j_3w})$.
Thus, we have
\begin{equation}
 \eta(\pi_{A}(\alpha^{i_3}\gamma^{j_3})\pi_{B}(\alpha^{i_3}\gamma^{j_3})\delta_C(\alpha^{i_3}\gamma^{j_3}))=\eta(\theta^{e_3}), \label{eq31}  
\end{equation}
where $e_3=\frac{r+1}{2}s-\frac{(r+1)w}{u}sj_3-\frac{v(r+1)}{u}\frac{s(s-1)}{2}+k^{\prime\prime}(r+1)+(\frac{q-1}{w}-1)i_3u-j_3w$.

Finally, for any $e\in A\cap B\cap C$, we will calculate $\eta(\delta_A(e)\delta_{B}(e) \delta_{C}(e))$. We  assume $e=\beta^{i_1}\alpha^{j_1}\in A$, where $0\leq i_1 \leq s-1$, $0\leq j_1\leq \frac{q-1}{u}-1$. Then by Eq. (\ref{eq29}), we have 
\begin{equation*}
 \delta_A(e)=\delta_A(\beta^{i_1}\alpha^{j_1})=\frac{q-1}{u}\theta^{e_1}. 
\end{equation*}
Similarly, we can assume $e\in B$ or $e\in C$. Thus,
$$\eta(\delta_{B}(e))
=\eta\left(\theta^{(\frac{q-1}{v}-1)i_2w-j_2v}\right), \ \eta(\delta_{C}(e))=
\eta(\theta^{(\frac{q-1}{w}-1)i_3u-j_3w}).$$
Hence, we have
\begin{equation}
\eta(\delta_A(e)\delta_{B}(e) \delta_{C}(e))=\eta(\theta^{e_4}),\label{eq32}
\end{equation}
where $e_4=e_1+(\frac{q-1}{v}-1)i_2w-j_2v+(\frac{q-1}{w}-1)i_3u-j_3w$.

By the conditions (2) and (3), we can deduce that $\eta(\theta^{e_i})=\eta(\theta^{\frac{v(r+1)}{u}\frac{s(s-1)}{2}})$ for any $i=1,2,3,4$. According to Theorem \ref{th1}, for any $e\in S$, we know that $\eta(\delta_S(e))$ are the same. When $n$ is even, by Lemma \ref{lem2}, there exists a $q$-ary MDS Euclidean self-dual code of length $n$. When $n$ is odd, by Lemma \ref{lem3}, there exists a $q$-ary MDS Euclidean self-dual code of length $n+1.$
\end{proof}
When we modify the conditions of Theorem \ref{th6} slightly, another construction of MDS Euclidean self-dual codes with length $n+1$ by Lemma \ref{lem3} is proposed as follows.
\begin{corollary}\label{cor2}
 Let $q=r^2$ and $r \equiv 3 \pmod 4$. Suppose $u$, $v$ and $w$ are three factors of $q-1$. Let $0\leq s\leq \frac{u}{\gcd(u,v)}$, $0 \leq t\leq \frac{v}{\gcd(v,w)}$ and $0\leq f\leq \frac{w}{\gcd(w,u)}$. Put
$$n=s\frac{q-1}{u}+t\frac{q-1}{v}+f\frac{q-1}{w}-\frac{2(q-1)st^{\prime}}{\operatorname{lcm}(u,v)}-\frac{2(q-1)tf^{\prime}}{\operatorname{lcm}(v,w)}-\frac{2(q-1)fs^{\prime}}{\operatorname{lcm}(w,u)}+\frac{4(q-1)s^{\prime}t^{\prime}f^{\prime}}{\operatorname{lcm}(u,v,w) }.$$
Suppose the following conditions hold: 

(1) $n$ is odd, $u$, $v$ and $w$ are even; 

(2) $u\mid (r+1)v $, $u\mid (r+1)w$, $v\mid (r-1)u $, $v\mid (r-1)w$, $w\mid (r-1)u $ and  $w\mid (r-1)v$;

(3)  both $\frac{(r+1)v}{u}s$ and $\frac{w(r+1)}{u}s$ are even; 

(4) $\frac{(r+1)v}{u}\frac{s(s-1)}{2}$ is even. \\
Then there exists a $q$-ary MDS Euclidean self-dual code of length $n+1$.   
\end{corollary}
\begin{proof}
Let $A$, $B$, $C$ and $S$ be defined as in Theorem \ref{th6}. From the proof of Theorem 6 and the condition (3), for any $e\in S$,    
$$\eta(\delta_S(e))=\eta(\theta^{\frac{(r+1)v}{u}\frac{s(s-1)}{2}})=1.$$ 
Since $\eta(-1)=1$, $\eta(-\delta_S(e))=1$. The theorem then follows from Lemma \ref{lem3}.
\end{proof}
Based on Corollary \ref{cor2}, when we consider adding the zero element into the set $A$, the MDS Euclidean self-dual codes with length $n+2$ are derived in the following.
\begin{corollary}\label{cor3}
 Let $q=r^2$ and $r \equiv 3 \pmod 4$. Suppose $u$, $v$ and $w$ are three factors of $q-1$. Let $0\leq s\leq \frac{u}{\gcd(u,v)}$, $0 \leq t\leq \frac{v}{\gcd(v,w)}$ and $0\leq f\leq \frac{w}{\gcd(w,u)}$. Put
$$n=s\frac{q-1}{u}+t\frac{q-1}{v}+f\frac{q-1}{w}-\frac{2(q-1)st^{\prime}}{\operatorname{lcm}(u,v)}-\frac{2(q-1)tf^{\prime}}{\operatorname{lcm}(v,w)}-\frac{2(q-1)fs^{\prime}}{\operatorname{lcm}(w,u)}+\frac{4(q-1)s^{\prime}t^{\prime}f^{\prime}}{\operatorname{lcm}(u,v,w) }.$$
Suppose the following conditions hold: 

(1) $n$, $u$, $v$ and $w$ are even; 

(2) $u\mid (r+1)v $, $u\mid (r+1)w$, $v\mid (r-1)u $, $v\mid (r-1)w$, $w\mid (r-1)u $ and  $w\mid (r-1)v$;

(3) both $\frac{(r+1)v}{u}s$ and $\frac{w(r+1)}{u}s$ are even;

(4) $\frac{(r+1)v}{u}\frac{s(s-1)}{2}$ is even. \\
Then there exists a $q$-ary MDS Euclidean self-dual code of length $n+2$.   
\end{corollary}
\begin{proof}
Please see \ref{app C} for the proof. 
\end{proof}

\section{Comparisons} 
In this section, we compare our results with previous ones. Some known results on MDS Euclidean self-dual codes are listed in Table \ref{tab D.3} (see \ref{appD}), while our results on the constructions of MDS Euclidean self-dual codes of length $n$ are presented in Table \ref{tab1}. Additionally, we provide some examples to illustrate that our constructions of MDS Euclidean self-dual codes are new.

\newcommand{\tabincell}[2]{\begin{tabular}
{@{}#1@{}}#2\end{tabular}}
{\footnotesize
\begin{longtable}{c|c|c|c}
\caption{Our results on MDS Euclidean self-dual codes of length $n$}\label{tab1}\\
\toprule
Classes & $q$ & $n$ & References\\
\midrule
\endfirsthead
\multicolumn{4}{c}{{\tablename\ \thetable{}  Our results on MDS Euclidean self-dual codes of length $n$ (continue)}} \\
\toprule
Classes & $q$ & $n$ & References\\
\midrule
\endhead
\bottomrule
\multicolumn{4}{r}{{}} \\
\endfoot
\bottomrule
\endlastfoot
\hline
1 & $q=r^{2}$, $r \equiv 3 \pmod 4$ & \tabincell{c}{$n=sl+(l_1+l_2)(r+1)+2$, \\$l$, $s$ even, $l\mid (r-1)$, \\ $0\leq s\leq \frac{r-1}{l}-1$,
$0\leq l_1\leq \frac{l}{2}$, $0\leq l_2\leq \frac{l}{2}$ } & Theorem \ref{th2}\\
\hline
2 & $q=r^{2}$, $r \equiv 3 \pmod 4$ & \tabincell{c}{$n=(s+1)l+(l_1+l_2)(r+1)-2l_2+2$, \\$l$, $s$ even, $l\mid (r-1)$, \\ $0\leq s\leq \frac{r-1}{l}-1$,
$0\leq l_1\leq \frac{l}{2}$, $0\leq l_2\leq \frac{l}{2}$ } & Theorem \ref{th3}\\
\hline
3& $q=r^{2}$, $r \equiv 3 \pmod 4$ &\tabincell{c}{$n=(s+s^{\prime})\frac{q-1}{u}+t\frac{q-1}{v}-2\frac{(q-1)\gcd(u,v)}{uv}st$,\\
$u$, $v$ even, $4\nmid v$, $\frac{(r+1)v}{2u}$ odd, \\ $2u\mid(r+1)v$, $v\mid (r-1)u$, $s+s^{\prime} $ odd, \\
$0\leq s, s^{\prime}\leq \frac{u}{\gcd(u,v)}$, $0\leq t\leq \frac{v}{\gcd(u,v)}$ }& Theorem \ref{th4}\\
\hline
4 &$q=r^{2}$, $r \equiv 3 \pmod 4$ &\tabincell{c}{$n=(s+s^{\prime})\frac{q-1}{u}+t\frac{q-1}{v}-2\frac{(q-1)\gcd(u,v)}{uv}st+2$,\\
$u$, $v$ even, $4\nmid v$, $2u\mid(r+1)v$, $v\mid (r-1)u$,\\  $0\leq s, s^{\prime}\leq \frac{u}{\gcd(u,v)}$, $0\leq t\leq \frac{v}{\gcd(u,v)}$, \\$\frac{(r+1)v}{2u}$ is even or both $s$ and $s^{\prime} $ are even }& Corollary \ref{cor1}\\
\hline
5 & $q=r^{2}$, $r \equiv 3 \pmod 4$ &\tabincell{c}{$n=(s+s^{\prime})\frac{q-1}{u}+t\frac{q-1}{v}-2\frac{(q-1)\gcd(u,v)}{uv}st+1$, \\
$n$ even, $4\nmid v$,
$2u\mid(r+1)v$, $v\mid (r-1)u$,\\$2^{l}\mid u$, $v\equiv 2^{l} \pmod {2^{l+1}}$ for some $l\geq 2$, \\ $0\leq s, s^{\prime}\leq \frac{u}{\gcd(u,v)}$, $0\leq t\leq \frac{v}{\gcd(u,v)}$, \\  $\frac{(r+1)v}{2u}$ is even or both $s$ and $s^{\prime} $ are even
}& Theorem \ref{th5}\\
\hline
6& $q=r^{2}$, $r \equiv 3 \pmod 4$ &\tabincell{c}{$n=s\frac{q-1}{u}+t\frac{q-1}{v}+f\frac{q-1}{w} -2\frac{(q-1)st^{\prime}}{\operatorname{lcm}(u,v)}-2\frac{(q-1)tf^{\prime}}{\operatorname{lcm}(v,w)}$\\$-2\frac{(q-1)fs^{\prime}}{\operatorname{lcm}(w,u)}+4\frac{(q-1)s^{\prime}t^{\prime}f^{\prime}}{\operatorname{lcm}(u,v,w)}$,
$n$, $u$, $v$ and $w$ even, \\
$u\mid (r+1)v$, $u\mid (r+1)w$, $v\mid (r-1)u$, \\$v\mid (r-1)w,$
$w\mid (r-1)u$ and $w\mid (r-1)v$,\\ $\frac{(r+1)v}{u}s$ and $\frac{(r+1)w}{u}s$ even}& Theorem \ref{th6}\\
\hline
7& $q=r^{2}$, $r \equiv 3 \pmod 4$ &\tabincell{c}{$n=s\frac{q-1}{u}+t\frac{q-1}{v}+f\frac{q-1}{w}-2\frac{(q-1)st^{\prime}}{\operatorname{lcm}(u,v)}-2\frac{(q-1)tf^{\prime}}{\operatorname{lcm}(v,w)}$\\$-2\frac{(q-1)fs^{\prime}}{\operatorname{lcm}(w,u)}+4\frac{(q-1)s^{\prime}t^{\prime}f^{\prime}}{\operatorname{lcm}(u,v,w)}+1$,
$n$, $u$, $v$ and $w$ even, \\
$u\mid (r+1)v$, $u\mid (r+1)w$, $v\mid (r-1)u$, \\$v\mid (r-1)w,$
$w\mid (r-1)u$ and $w\mid (r-1)v$, \\$\frac{(r+1)v}{u}s$, $\frac{(r+1)w}{u}s$ and  $\frac{(r+1)v}{u}\frac{s(s-1)}{2}$ even}& Corollary \ref{cor2}\\
\hline
8& $q=r^{2}$, $r \equiv 3 \pmod 4$ &\tabincell{c}{$n=s\frac{q-1}{u}+t\frac{q-1}{v}+f\frac{q-1}{w}-2\frac{(q-1)st^{\prime}}{\operatorname{lcm}(u,v)}-2\frac{(q-1)tf^{\prime}}{\operatorname{lcm}(v,w)}$\\$-2\frac{(q-1)fs^{\prime}}{\operatorname{lcm}(w,u)}+4\frac{(q-1)s^{\prime}t^{\prime}f^{\prime}}{\operatorname{lcm}(u,v,w)}+2$,
$n$, $u$, $v$ and $w$ even, \\$u\mid (r+1)v$, $u\mid (r+1)w$, $v\mid (r-1)u$,\\ $v\mid (r-1)w,$
$w\mid (r-1)u$ and $w\mid (r-1)v$, \\$\frac{(r+1)v}{u}s$, $\frac{(r+1)w}{u}s$ and $\frac{(r+1)v}{u}\frac{s(s-1)}{2}$ even}& Corollary \ref{cor3}\\
\hline
\end{longtable}
}
To compare our results with previous ones more intuitively  , we divide the known constructions of MDS Euclidean self-dual codes of length $n$ into the following cases.
\begin{itemize}
    \item $n\leq r$: Classes 3 and 5;
    \item $p\mid n$: Classes 4, 8, 10, 14 and 15;
    \item $p\mid (n-1)$: Classes 2, 9, 11, 16 and 17;
    \item $n=tm+c$ ($c=0,1,2$) and $m\mid (q-1)$: Classes 6, 7, 9, 12, 13, 18, 20, 21, 22, 23, 24, 25, 26, 27, 28, 29 and 30;
    \item $n=sa+tb+c$ ($c$ can be negative) and $a,b\mid (q-1)$: Classes 31 and 32, 33, 34, 35, 36, 37, 38, 39, 40, 41, 42, 43, 44 and 45;
\end{itemize}
(1) With the above classification, we can conclude that Classes 1-5 of Table \ref{tab1} are likely to be included only in the case of length $n=sa+tb+c$. However, the constant $c$ in our constructions is more flexible, which may provide more possibilities for the length. It is easy to verify that Classes 2, 3, 4 and 5 in Table \ref{tab1} can give rise to many new $q$-ary MDS Euclidean self-dual codes. Furthermore, the other classes in Table \ref{tab1} can increase the proportion of MDS Euclidean self-dual codes we have constructed out of all possible MDS Euclidean self-dual codes.\\ 
(2) Based on the above classification, it is clear that Classes 6-8 of Table \ref{tab1} are not included in any previous results.\\
(3) In \cite{FXF21}, for the fixed $q$,  MDS Euclidean self-dual codes with larger lengths are missing more, however, our results can complement them somewhat. Furthermore, under certain conditions, when $s = 0$, some of the results in \cite{FXF21} are special cases of our results.\\
(4) In Table \ref{tab2}, we list the proportion of possible lengths relative to $\frac{q}{2}$, where $N$ is the number of all MDS Euclidean self-dual codes constructed in each reference and $N_1$ is the number of new lengths in our constructions. The constructions in \cite{HFF21} and \cite{WLZ23} account for more than $34\%$ and $38\%$ of all possible MDS Euclidean self-dual codes. And in \cite{FXF21}, the constructions can give  more than $56\%$ of all possible MDS Euclidean self-dual codes. Additionally, 
all the constructions in Table \ref{tab D.3} can contribute more than $64\%$  of such codes in total.  However, in our results, $N/\frac{q}{2}$ is more than $85\%$, which is a considerable improvement compared to the previous results.

{\footnotesize
\begin{longtable}{c|c|c|c|c|c|c}
\caption{The ratios of $N$ and $\frac{q}{2}$}\label{tab2}\\
\toprule
$r$ & 
$N/(\frac{q}{2})$ of \cite{HFF21}&  $N/(\frac{q}{2})$ of \cite{WLZ23} & $N/(\frac{q}{2})$ of \cite{FXF21} &\tabincell{c}{ $N/(\frac{q}{2})$ of Table \ref{tab D.3}}&  $N/(\frac{q}{2})$ of us & \tabincell{c}{$N_1$}\\
\midrule
\endfirsthead
\multicolumn{7}{c}{{\tablename\ \thetable{}  xushangye}} \\
\toprule
$r$ & 
$N/(\frac{q}{2})$ of [18]&  $N/(\frac{q}{2})$ of [19] & $N/(\frac{q}{2})$ of [20] & \tabincell{c}{$N/(\frac{q}{2})$ of Table \ref{tab D.3}}& $N/(\frac{q}{2})$ of us & $N_1$\\
\midrule
\endhead
\bottomrule
\multicolumn{7}{r}{{continue}} \\
\endfoot
\bottomrule
\endlastfoot
\hline
151&
$34.95\%$ &$38.64\%$ & $56.32\%$ & $64.15\%$ &  $85.1\%$& 2459\\
\hline
163&
$34.28\%$& $38.56\%$& $56.67\%$& $63.99\%$& $85.19\%$& 2856\\
\hline
167 &
$34.27\%$&$38.53\%$ &$56.09\%$&$63.69\%$ & $85.19\%$&3031\\
\hline
179 &$34.9\%$ & $38.46\%$&$56.25\%$ &$64.3\%$ & $85.26\%$& 3448\\
\hline
\end{longtable}
}
Finally, we present some examples of our new constructions as follows.
\begin{example}
  In Theorem \ref{th3}, let $p=r=19$, $q=r^2=361$, $s=0$, $l=18$, $l_1=8$ and $l_2=6$. Then $n=(s+1)l+(l_1+l_2)(r+1)-2l_2+2=288$. So we can get an MDS Euclidean self-dual codes of length $288$ over $\mathbb{F}_{361}$. 
\end{example}
\begin{example}
  In Theorem \ref{th4}, let $p=r=19$, $q=r^2=361$, $s=9$, $s^{\prime}=10$, $t=2$, $ u=20$, and $v=18$. Then $n=(s+s^{\prime})\frac{q-1}{u}+t\frac{q-1}{v}-2\frac{(q-1)\gcd(u,v)}{uv}st=310$. So we can get an MDS Euclidean self-dual codes of length $310$ over $\mathbb{F}_{361}$. 
\end{example}
\begin{example}
  In Corollary \ref{cor1}, let $p=r=19$, $q=r^2=361$, $s=2$, $s^{\prime}=10$, $t=1$, $ u=20$, and $v=18$. Then $n=(s+s^{\prime})\frac{q-1}{u}+t\frac{q-1}{v}-2\frac{(q-1)\gcd(u,v)}{uv}st+2=314$. So we can get an MDS Euclidean self-dual codes of length $314$ over $\mathbb{F}_{361}$. 
\end{example}
\begin{example}
  In Theorem \ref{th5}, let $p=r=19$, $q=r^2=361$, $s=2$, $s^{\prime}=4$, $t=7$, $ u=20$, and $v=36$. Then $n=(s+s^{\prime})\frac{q-1}{u}+t\frac{q-1}{v}-2\frac{(q-1)\gcd(u,v)}{uv}st+2=124$. So we can get an MDS Euclidean self-dual codes of length $124$ over $\mathbb{F}_{361}$. 
\end{example}
\begin{example}
 In Theorem \ref{th6}, let $p=r=151$, $q=r^2=22801$, $u=164$, $v=2$, $w=108$, $s=3$, $t=1$ and $f=13$. Then $n=s\frac{q-1}{u}+t\frac{q-1}{v}+f\frac{q-1}{w}-\frac{2(q-1)st^{\prime}}{\operatorname{lcm}(u,v)}-\frac{2(q-1)tf^{\prime}}{\operatorname{lcm}(v,w)}-\frac{2(q-1)fs^{\prime}}{\operatorname{lcm}(w,u)}+\frac{4(q-1)s^{\prime}t^{\prime}f^{\prime}}{\operatorname{lcm}(u,v,w) }=9912$. So we can get an MDS Euclidean self-dual codes of length $9912$ over $\mathbb{F}_{22801}$.
\end{example}
\begin{example}
 In Corollary \ref{cor3}, let $p=r=151$, $q=r^2=22801$, $u=152$, $v=2$, $w=20$, $s=11$, $t=1$ and $f=2$. Then $n=s\frac{q-1}{u}+t\frac{q-1}{v}+f\frac{q-1}{w}-\frac{2(q-1)st^{\prime}}{\operatorname{lcm}(u,v)}-\frac{2(q-1)tf^{\prime}}{\operatorname{lcm}(v,w)}-\frac{2(q-1)fs^{\prime}}{\operatorname{lcm}(w,u)}+\frac{4(q-1)s^{\prime}t^{\prime}f^{\prime}}{\operatorname{lcm}(u,v,w) }=8192$. So we can get an MDS Euclidean self-dual codes of length $8192$ over $\mathbb{F}_{22801}$. 
\end{example}

\section{Conclusion}
 In this paper, we construct six new classes of MDS Euclidean self-dual codes over $\mathbb{F}_q$ via (extended) GRS codes (see Theorems \ref{th3}, \ref{th4}, \ref{th5}, \ref{th6} and Corollaries \ref{cor1}, \ref{cor3}). When these results are combined with those in Table \ref{tab1}, the proportion of all known MDS Euclidean self-dual codes relative to the total number of possible MDS Euclidean self-dual codes exceeds $85\%$ (see Table \ref{tab2}). The main idea of our construction is to choose suitable evaluation sets such that the corresponding (extended) GRS codes are Euclidean self-dual codes. It is worth mentioning that we provide a method for choosing an evaluation set from multiple intersecting subsets of finite fields. 
Additionally, we present a useful theorem (see Theorem \ref{th1}) that guarantees the chosen evaluation set satisfies the conditions in Lemma \ref{lem2} or \ref{lem3}, leading to new MDS Euclidean self-dual codes. It can be seen that the ratio $N/\frac{q}{2}$ increases significantly when a suitable evaluation set is chosen from three subsets. Therefore, constructing MDS Euclidean self-dual codes via more than three subsets of finite field is one of our future research directions.
\appendix
\section{The proof of Theorem 3}\label{app A}
\begin{proof}
 Let $$A=\bigcup\limits_{k=1}^{s}H_k, B=\bigcup\limits_{i=0}^{l_1-1} N_{2i+1}, C=\bigcup\limits_{j=0}^{l_2-1}N_{2j}, A^{\prime}=A\cup H_0, B^{\prime}=B\cup\{0\},$$ where $H_k$ and $N_i$ are defined in Eq. (\ref{eq6}) and Eq. (\ref{eq7}), respectively. Choose $$S=(A^{\prime}\backslash (B^{\prime}\cup C))\bigcup (B^{\prime}\backslash (A^{\prime}\cup C))\bigcup (C\backslash (A^{\prime}\cup B^{\prime}))\bigcup(A^{\prime}\cap B^{\prime}\cap C),$$   
 then by Lemma \ref{lem6},
  $$|S|=|A^{\prime}|+|B^{\prime}|+|C|-|A^{\prime}\cap C|=(s+1)l+(l_1+l_2)(r+1)-2l_2+1=n.$$
  
Firstly, for any $a\in A^{\prime}\backslash (B^{\prime}\cup C)$, we will calculate $\eta(\delta_{A^{\prime}}(a)\pi_{B^{\prime}}(a)\pi_{C}(a))$. Assume that $a\in H_k$, where $0\leq k\leq s$. Then by Lemma \ref{lem1} and Lemma \ref{lem7},
 \begin{eqnarray*}
     \delta_{A^{\prime}}(a)\pi_{B^{\prime}}(a)\pi_{C}(a)&=& \pi_{H_k}^{\prime}(a)\prod_{m=0, m\neq k}^{s}\pi_{H_m}(a)\left(\prod_{i=0}^{l_1-1}\pi_{N_{2i+1}}(a)\cdot a\right)\prod_{j=0}^{l_2-1}\pi_{N_{2j}}(a)\\
     &=&a \cdot \pi_{H_k}^{\prime}(a)\prod_{m=0, m\neq k}^{s}\pi_{H_m}(a)\prod_{i=0}^{l_1-1}(N(a)-\alpha^{2i+1})\prod_{j=0}^{l_2-1}(N(a)-\alpha^{2j})\in \mathbb{F}_r^{*}.
\end{eqnarray*}
Hence, $\eta\left(\delta_{A^{\prime}}(a)\pi_{B^{\prime}}(a)\pi_{C}(a)\right)=1$.

Secondly, for any $b\in B^{\prime} \backslash (A^{\prime}\cup C)$, we will  calculate $\eta(\pi_{A^{\prime}}(b)\delta_{B^{\prime}}(b)\pi_{C}(b))$. If $b=0$, then
$$\pi_{A^{\prime}}(0)\delta_{B^{\prime}}(0)\pi_{C}(0)=\pi_{H_0}(0)\pi_{A}(0)\delta_{B^{\prime}}(0)\pi_{C}(0).$$
By $\pi_{H_0}(0)\in \mathbb{F}_r^{*}$ and Eq. (13), we obtain $\eta\left( \pi_{A^{\prime}}(0)\delta_{B^{\prime}}(0)\pi_{C}(0)\right)=\eta\left( \pi_{A}(0)\delta_{B^{\prime}}(0)\pi_{C}(0)\right)=1.$
If $b\neq0$, we can assume  $b\in N_{2i+1}$, where $0\leq i\leq l_1-1$, then 
\begin{equation*}
\pi_{A^{\prime}}(b)\delta_{B^{\prime}}(b)\pi_{C}(b)= \pi_{H_0}(b)\pi_{A}(b)\delta_{B^{\prime}}(b)\pi_{C}(b).    
\end{equation*}
In Lemma \ref{lem8}, we have proved that $\pi_{H_0}(b)$ is a square element in $\mathbb{F}_q$. According to Eq. (\ref{eq14}), we have
$\eta\left(\pi_{A^{\prime}}(b)\delta_{B^{\prime}}(b)\pi_{C}(b)\right)=\eta\left(\pi_{A}(b)\delta_{B^{\prime}}(b)\pi_{C}(b) \right)=1.$

Finally, for any $c\in C \backslash (A^{\prime}\cup B^{\prime})$, we will calculate $\eta(\pi_{A^{\prime}}(c)\pi_{B^{\prime}}(c)\delta_{C}(c))$. Assume $c\in N_{2j}$, where $0\leq j\leq l_2-1$, then
$$\pi_{A^{\prime}}(c)\pi_{B^{\prime}}(c)\delta_{C}(c)=\pi_{H_0}(c)\pi_{A}(c)\pi_{B^{\prime}}(c)\delta_{C}(c).$$
Similarly as above, we can show that 
$\eta\left(\pi_{A^{\prime}}(c)\pi_{B^{\prime}}(c)\delta_{C}(c) \right)=\eta\left(\pi_{A}(c)\pi_{B^{\prime}}(c)\delta_{C}(c) \right)=1$.

In summary, for any $a\in A^{\prime} \backslash (B^{\prime}\cup C)$, $b\in B^{\prime} \backslash (A^{\prime}\cup C)$ and $c\in C \backslash (A^{\prime}\cup B^{\prime})$, we can deduce that
$$\eta\left(\delta_{A^{\prime}}(a)\pi_{B^{\prime}}(a)\pi_{C}(a)\right)=\eta\left(\pi_{A^{\prime}}(b)\delta_{B^{\prime}}(b)\pi_{C}(b)\right)=\eta\left(\pi_{A^{\prime}}(c)\pi_{B^{\prime}}(c)\delta_C(c)\right)=1.$$   
By Theorem \ref{th1}, $\eta\left(-\delta_S(e)\right)=\eta\left(\delta_S(e)\right)=1$ for any $e\in S$. Again by Lemma \ref{lem3}, there exists a $q$-ary MDS Euclidean self-dual code of length $n+1$.  
\end{proof}
\section{The proof of Corollary 1}\label{app B}
\begin{proof}
Choose $$S=\left(A^{\prime}\backslash \left(B\cup C\right)\right)\bigcup \left(B\backslash \left(A^{\prime}\cup C\right)\right)\bigcup \left(C\backslash \left(A^{\prime}\cup B\right)\right)\bigcup\left(A^{\prime}\cap B\cap C\right),$$ 
where $A$, $B$ and $C$ are defined as in Eq. (\ref{eq15}) and $A^{\prime}=A\cup\{0\}$.
Hence, 
$|S|=|A|+|B|+|C|-2|A\cap B|+1=n+1$.  
By Lemma \ref{lem1},
\begin{eqnarray*}
\delta_{A^{\prime}}(0)\pi_B(0)\pi_C(0)&=&\prod_{i=0}^{s-1}\pi_{A_i}(0)\prod_{j=0}^{t-1}\pi_{B_j}(0)\prod_{k=0}^{s^{\prime}-1}\pi_{C_k}(0)\\
&=&(-1)^{s+t+s^{\prime}}\theta^{\frac{(r+1)v}{2u}{s^{\prime}}^2(r-1)}\prod_{i=0}^{s-1}\beta^{\frac{q-1}{u}i}\prod_{j=0}^{t-1}\alpha^{\frac{q-1}{v}j}.
\end{eqnarray*}
Since $r-1$ is even, we have
$\eta\left(\delta_{A^{\prime}}(0)\pi_B(0)\pi_C(0)\right)=1$.
For any $\beta^i\alpha^j\in A^{\prime} \backslash (B \cup C)$, where $0\leq i \leq s-1$, $0\leq j\leq \frac{q-1}{u}-1$,
$$\delta_{A^{\prime}}(\beta^i\alpha^j)\pi_B(\beta^i\alpha^j)\pi_C(\beta^i\alpha^j)=\beta^i\alpha^j\cdot\delta_{A}(\beta^i\alpha^j)\pi_B(\beta^i\alpha^j)\pi_C(\beta^i\alpha^j).$$
By the condition (1) and Eq. (\ref{eq17}), 
\begin{equation*}
   \eta\left(\delta_{A^{\prime}}(\beta^i\alpha^j)\pi_B(\beta^i\alpha^j)\pi_C(\beta^i\alpha^j)\right)=\eta\left(\theta^{\frac{(r+1)v}{2u}s^{\prime}}\right). 
\end{equation*}
For any $\alpha^i\beta^j\in B \backslash (A^{\prime} \cup C)$, where $0\leq i \leq t-1$, $0\leq j\leq \frac{q-1}{v}-1$, by Lemma \ref{lem1},
$$\pi_{A^{\prime}}(\alpha^i\beta^j)\delta_B(\alpha^i\beta^j)\pi_C(\alpha^i\beta^j)=\alpha^i\beta^j\cdot\pi_A(\alpha^i\beta^j)\delta_B(\alpha^i\beta^j)\pi_C(\alpha^i\beta^j),$$
Similarly, according to Eq. (\ref{eq18}), we can deduce that
\begin{equation*}
   \eta\left(\pi_{A^{\prime}}(\alpha^i\beta^j)\delta_B(\alpha^i\beta^j)\pi_C(\alpha^i\beta^j)\right)=\eta\left(\theta^{\frac{(r+1)v}{2u}s^{\prime}}\right). 
\end{equation*}
For any $\gamma^{2i+1}\alpha^j\in C \backslash (A \cup B)$, where $0\leq i \leq s^{\prime}-1$, $0\leq j\leq \frac{q-1}{u}-1$, by Lemma \ref{lem1},
\begin{equation*}
\pi_{A^{\prime}}(\gamma^{2i+1}\alpha^j)\pi_B(\gamma^{2i+1}\alpha^j)\delta_C(\gamma^{2i+1}\alpha^j)=\gamma^{2i+1}\alpha^j\cdot\pi_{A}(\gamma^{2i+1}\alpha^j)\pi_B(\gamma^{2i+1}\alpha^j)\delta_C(\gamma^{2i+1}\alpha^j).    
\end{equation*}
By Eq. (\ref{eq19}) and $\eta(\gamma)=-1$, we can show that
 \begin{equation*}
 \eta\left(\pi_{A^{\prime}}(\gamma^{2i+1}\alpha^j)\pi_B(\gamma^{2i+1}\alpha^j)\delta_C(\gamma^{2i+1}\alpha^j)\right)=\eta\left(\theta^{\frac{(r+1)v}{2u}s}\right).
 \end{equation*}
 
When $\frac{(r+1)v}{2u}$ is even or both $s$ and $s^{\prime}$ are even, we have $\eta\left(\theta^{\frac{(r+1)v}{2u}s^{\prime}}\right)=\eta\left(\theta^{\frac{(r+1)v}{2u}s}\right)=1$. According to Theorem \ref{th1}, for any $e\in S$, we have $\eta(-\delta_S(e))=\eta(\delta_S(e))=1$. Again by Lemma \ref{lem3}, there exists a $q$-ary MDS Euclidean self-dual code of length $n+2$. 
\end{proof}
\section{The proof of Corollary 3}\label{app C}
\begin{proof}
Choose 
\begin{equation*}
S=\left(A^{\prime}\backslash \left(B\cup C\right)\right)\bigcup \left(B\backslash \left(A^{\prime}\cup C\right)\right)\bigcup \left(C\backslash \left(A^{\prime}\cup B\right)\right)\bigcup\left(A^{\prime}\cap B\cap C\right),    
\end{equation*}  
where $A$, $B$ and $C$ are defined as in Theorem \ref{th6}, and  $A^{\prime}=A\cup \{0\}$. Then $|S|=n+1$. Firstly, by Lemma \ref{lem1},
\begin{eqnarray*}
\delta_{A^{\prime}}(0)\pi_B(0)\pi_C(0)&=&\prod\limits_{i=0}^{s-1}\pi_{A_i}(0)\prod\limits_{j=0}^{t-1}\pi_{B_j}(0)\prod_{l=0}^{f-1}\pi_{C_l}(0)\\
&=&(-1)^{s+t+l}\prod\limits_{i=0}^{s-1}\beta^{\frac{q-1}{u}i}\prod\limits_{j=0}^{t-1}\gamma^{\frac{q-1}{v}j}\prod\limits_{l=0}^{f-1}\alpha^{\frac{q-1}{w}l}.
\end{eqnarray*}
Since $u$, $v$ and $w$ are even,  
$$\eta(\delta_{A^{\prime}}(0)\pi_B(0)\pi_C(0))=1.$$
Secondly, for any $\beta^{i_1}\alpha^{j_1}\in A^{\prime}\backslash (B\cup C)$, where $0\leq i_1 \leq s-1$, $0\leq j_1\leq \frac{q-1}{u}-1$, 
$$\delta_{A^{\prime}}(\beta^{i_1}\alpha^{j_1})\pi_B(\beta^{i_1}\alpha^{j_1})\pi_c(\beta^{i_1}\alpha^{j_1})=\beta^{i_1}\alpha^{j_1}\delta_{A}(\beta^{i_1}\alpha^{j_1})\pi_B(\beta^{i_1}\alpha^{j_1})\pi_c(\beta^{i_1}\alpha^{j_1}),$$
from the condition (1) and Eq. (\ref{eq29}), we can deduce that
$$\eta\left(\delta_{A^{\prime}}(\beta^{i_1}\alpha^{j_1})\pi_B(\beta^{i_1}\alpha^{j_1})\pi_c(\beta^{i_1}\alpha^{j_1})\right)=\eta(\theta^{e_1}).$$
For any $\gamma^{i_2}\beta^{j_2}\in B\backslash(A^{\prime}\cup C)$ and $\alpha^{i_3}\gamma^{j_3}\in C\backslash (A^{\prime}\cup B)$, where $0\leq i_2\leq t-1$, $0\leq j_2\leq \frac{q-1}{v}-1$, $0\leq i_3\leq f-1$ and $0\leq j_3\leq \frac{q-1}{w}-1$, 
$$\pi_{A^{\prime}}(\gamma^{i_2}\beta^{j_2})\delta_B(\gamma^{i_2}\beta^{j_2})\pi_C(\gamma^{i_2}\beta^{j_2})=\gamma^{i_2}\beta^{j_2}\pi_{A}(\gamma^{i_2}\beta^{j_2})\delta_B(\gamma^{i_2}\beta^{j_2})\pi_C(\gamma^{i_2}\beta^{j_2}),$$
$$\pi_{A^{\prime}}(\alpha^{i_3}\gamma^{j_3})\pi_B(\alpha^{i_3}\gamma^{j_3})\delta_C(\alpha^{i_3}\gamma^{j_3})=\alpha^{i_3}\gamma^{j_3}\pi_{A}(\alpha^{i_3}\gamma^{j_3})\pi_B(\alpha^{i_3}\gamma^{j_3})\delta_C(\alpha^{i_3}\gamma^{j_3}).$$
Thus, from the condition (1) and Eqs. (\ref{eq30}) and (\ref{eq31}), 
$$\eta(\pi_{A^{\prime}}(\gamma^{i_2}\beta^{j_2})\delta_B(\gamma^{i_2}\beta^{j_2})\pi_C(\gamma^{i_2}\beta^{j_2}))=\eta(\theta^{e_2}),$$
$$\eta(\pi_{A^{\prime}}(\alpha^{i_3}\gamma^{j_3})\pi_B(\alpha^{i_3}\gamma^{j_3})\delta_C(\alpha^{i_3}\gamma^{j_3}))=\eta(\theta^{e_3}).$$
Finally, for any $e\in A^{\prime}\cap B\cap C=A\cap B\cap C$, by Eq. \ref{eq32},
$\eta(\delta_A(e)\delta_{B}(e) \delta_{C}(e))=\eta(\theta^{e_4})$ can be obtained. 
In summary, by Theorem \ref{th1} and the conditions (1)-(4),  for any $e\in S$,
$$\eta(\delta_S(e))=\eta(-\delta_S(e))=1.$$
The theorem then follows from Lemma \ref{lem3}.
\end{proof}
\section{Some known results on MDS Euclidean self-dual codes of length $n$}\label{appD}

{\footnotesize
\begin{longtable}{c|c|c|c}
\caption{Some known results on MDS Euclidean self-dual codes of length $n$} \label{tab D.3}\\
\toprule
Classes & $q$ & $n$ &References\\
\midrule
\endfirsthead
\multicolumn{4}{c}{{\tablename\ \thetable{} Some known results on MDS Euclidean self-dual codes of length $n$ (continue) }} \\
\toprule
Classes & $q$ & $n$ &References\\
\midrule
\endhead
\bottomrule
\multicolumn{4}{r}{} \\
\endfoot
\bottomrule
\endlastfoot

 1& $q$ even  & $n \leq q$ & \cite{GG08} \\
 \hline
 2& $q$ odd  & $n = q+1$ & \cite{GG08,JX17} \\
  \hline
 3& $q=r^{2}$ & $n \leq r$ & \cite{JX17} \\
  \hline
  4&$q=r^{2}$, $r \equiv 3 \pmod 4$ & $n=2tr, t \leq \frac{r-1}{2}$ & \cite{JX17} \\
  \hline
  5&$q \equiv 1\pmod 4$ & $4^{n}n^{2} \leq q$ & \cite{JX17} \\
 \hline 
  6&$q \equiv 3\pmod 4$ & $n \equiv \textnormal{0 (mod 4)}$ and $(n-1) \mid (q-1)$ & \cite{TW17}\\
  \hline
 7& $q \equiv 1 \pmod  4$ & $(n-1) \mid (q-1)$ & \cite{TW17} \\
  \hline
 8& $q \equiv 1 \pmod  4$ & $n=2p^{\ell}, \ell \leq m$ & \cite{FF19} \\
  \hline
  9&$q \equiv 1 \pmod  4$ & $n=p^{\ell}+1, \ell \leq m$ & \cite{FF19} \\
  \hline
 10& $q=r^{s}$, $r$ odd, $s$ even & \tabincell{c}{$n=2tr^{\ell}$, $0 \leq \ell <s$,  and $1 \leq t \leq \frac{r-1}{2}$} & \cite{FF19} \\
  \hline
 11& $q=r^{s}$, $r$ odd, $s$ is even & \tabincell{c}{$n=(2t+1)r^{\ell}+1$, \\$0 \leq \ell <s$ and $0 \leq t \leq \frac{r-1}{2}$} & \cite{FF19} \\
  \hline
  12&$q$ odd & $(n-2) \mid (q-1)$, $\eta(2-n)=1$ & \cite{FF19,Y18} \\
  \hline
  13&$q \equiv 1 \pmod  4$ & $n \mid (q-1)$ & \cite{Y18} \\
  \hline
 14& $q=r^{s}$, $r$ odd, $s \geq 2$ & $n=tr$, $t$ even  and  $2t \mid (r-1)$ &\cite{Y18}\\
  \hline
  15&$q=r^{s}$, $r$ odd, $s \geq 2$ & \tabincell{c}{$n=tr$, $t$ even, \\$(t-1) \mid (r-1)$ and $\eta(1-t)=1$} &\cite{Y18}\\
  \hline
  16&$q=r^{s}$, $r$ odd, $s \geq 2$ & \tabincell{c}{$n=tr+1$, $t$ odd, $t \mid (r-1)$ and $\eta(t)=1$} &\cite{Y18}\\
  \hline
  17&$q=r^{s}$, $r$ odd, $s \geq 2$ & \tabincell{c}{$n=tr+1$, $t$ odd, $(t-1) \mid (r-1)$, \\ $\eta(t-1)=\eta(-1)=1$} &\cite{Y18}\\
  \hline
  18&$q=r^{2}$, $r$ odd & \tabincell{c}{$n=tm$, $\frac{q-1}{m}$ even,  and $1 \leq t \leq \frac{r+1}{\gcd(r+1, m)}$} & \cite{FLLL20}\\
  \hline
  19&$q=r^{2}$, $r$ odd & \tabincell{c}{$n=tm+1$, $tm$ odd, $m \mid (q-1)$,  \\ $2 \leq t \leq \frac{r+1}{2\gcd(r+1, m)}$} & \cite{FLLL20}\\
  \hline
 
  20&$q=r^{2}$, $r$ odd & \tabincell{c}{$n=tm+2$, $m \mid (q-1)$, $tm$ even, \\and $1 \leq t \leq \frac{r+1}{\gcd(r+1, m)}$ \\(except $t, m$ even and $r \equiv 1 \pmod 4$),   } & \cite{FLLL20}\\
  \hline
  21&$q=r^{2}$, $r$ odd & \tabincell{c}{$n=tm$, $\frac{q-1}{m}$ even, $1 \leq t \leq \frac{s(r-1)}{\gcd(s(r-1), m)}$,\\ $s$ even, $s \mid m$, and $\frac{r+1}{s}$ even } & \cite{FLLL20}\\
  \hline
 22& $q=r^{2}$, $r$ odd & \tabincell{c}{$n=tm+2$, $\frac{q-1}{m}$ even, $s$ even and $s \mid m$, \\$1 \leq t \leq \frac{s(r-1)}{\gcd(s(r-1), m)}$,   $\frac{r+1}{s}$ even } & \cite{FLLL20}\\
  \hline
  23&$q=r^{2}$, $r$ odd & \tabincell{c}{$n=tm$, $\frac{q-1}{m}$ even,  \\and $2 \leq t \leq \frac{r-1}{\gcd(r-1, m)}$} & \cite{LLL19}\\
  \hline
  24&$q=r^{2}$, $r$ odd & \tabincell{c}{$n=tm+1$, $tm$ odd, $m \mid (q-1)$,  \\and $2 \leq t \leq \frac{r-1}{\gcd(r-1, m)}$} & \cite{LLL19}\\
  \hline
  25&$q=r^{2}$, $r$ odd & \tabincell{c}{$n=tm+2$, $tm$ even, $m \mid (q-1)$, \\and $2 \leq t \leq \frac{r-1}{\gcd(r-1, m)}$} & \cite{LLL19}\\
  \hline
  26&$q=r^{2}$, $r$ odd & \tabincell{c}{$n=tm$, $tm$ even, $m \mid (q-1)$, $1 \leq t \leq \frac{r-1}{n_2}$,\\ $\frac{r+1}{\gcd(r+1,m)}$ even, $n_2=\frac{m}{\gcd(r+1,m)}$ } & \cite{FZXF22}\\
  \hline
  27& $q=r^{2}$, $r$ odd & \tabincell{c}{$n=tm+1$, $tm$ odd, $m \mid (q-1)$, \\$n_2=\frac{m}{\gcd(r+1,m)}$  and $1 \leq t \leq \frac{r-1}{n_2}$} & \cite{FZXF22}\\
  \hline
  28&$q=r^{2}$, $r$ odd & \tabincell{c}{$n=tm+2$, $tm$ even, $m \mid (q-1)$,\\ $n_2=\frac{m}{\gcd(r+1,m)}$  and $1 \leq t \leq \frac{r-1}{n_2}$} & \cite{FZXF22}\\
  \hline
  29&$q=r^{2}$, $r$ odd & \tabincell{c}{$n=tm$, $m \mid (q-1)$ , and $1 \leq t \leq \frac{r+1}{n_2}$,  \\$n_1=\gcd(r-1,m)$,  $n_2=\frac{m}{n_1}$, \\ $\frac{q-1}{m}$ and $\frac{r-1}{n_1}+tn_2$ even } & \cite{FZXF22}\\
  \hline
  30&$q=r^{2}$, $r$ odd & \tabincell{c}{$n=tm+2$, $m \mid (q-1)$ , $1 \leq t \leq \frac{r+1}{n_2}-1$,  \\$n_1=\gcd(r-1,m)$,  $n_2=\frac{m}{n_1}$,  \\$\frac{(r+1)(t-1)}{2}$ even, or $n_2$ odd and $t$ even } & \cite{FZXF22}\\
  \hline
  31& $q=r^{2}$, $r \equiv 1  \pmod4$ & \tabincell{c}{$n=s(r-1)+t(r+1)$, $s$ even, \\$1 \leq s \leq \frac{r+1}{2}$ $1 \leq t \leq \frac{r-1}{2}$} & \cite{FLL21} \\
\hline  
  32&$q=r^{2}$, $r \equiv 3 \pmod 4$ & \tabincell{c}{$n=s(r-1)+t(r+1)$, $s$ odd,\\ $1 \leq s \leq \frac{r+1}{2}$ $1 \leq t \leq \frac{r-1}{2}$} & \cite{FLL21} \\
\hline  
  33&$q=r^{2}$, $r \equiv 1 \pmod 4$ & \tabincell{c} {$n=s\frac{q-1}{a}+t\frac{q-1}{b}$, $b$ and $s$ are even, \\$2a\mid b(r+1)$, $2b\mid a(r-1)$, $a \equiv 2 \pmod 4$,  \\ $1\leq s\leq \frac{a}{\gcd(a,b) }$, $1\leq t\leq \frac{b}{\gcd(a,b) }$ } & \cite{HFF21}\\
  \hline
  34& $q=r^{2}$, $r \equiv 3 \pmod 4$ & \tabincell{c} {$n=s\frac{q-1}{a}+t\frac{q-1}{b}$, $a$ even, $\frac{(r+1)b}{2a}s^2$ odd, \\$2a\mid b(r+1)$, $2b\mid a(r-1)$, $b \equiv 2 \pmod 4$,\\
  $1\leq s\leq \frac{a}{\gcd(a,b) }$, $1\leq t\leq \frac{b}{\gcd(a,b) }$}& \cite{HFF21}\\
  \hline
  35&$q=r^{2}$, $r \equiv 1 \pmod 4$ & \tabincell{c} {$n=s\frac{q-1}{a}+t\frac{q-1}{b}+2$, $b$ even, $s$ odd, \\$2a\mid b(r+1)$, $2b\mid a(r-1)$, $a \equiv 2 \pmod 4$, \\$1\leq s\leq \frac{a}{\gcd(a,b) }$, $1\leq t\leq \frac{b}{\gcd(a,b) }$ } & \cite{HFF21}\\
  \hline
36&$q=r^{2}$, $r \equiv 3 \pmod 4$ & \tabincell{c} {$n=s\frac{q-1}{a}+t\frac{q-1}{b}+2$, $a$ even,  $2a\mid b(r+1)$,\\ $2b\mid a(r-1)$, $b \equiv 2 \pmod 4$, 
$\frac{(r+1)b}{2a}s^2$ even, \\$1\leq s\leq \frac{a}{\gcd(a,b) }$, $1\leq t\leq \frac{b}{\gcd(a,b) }$}& \cite{HFF21}\\
\hline
37&$q=r^{2}$, $r \equiv 3 \pmod 4$ & \tabincell{c}{$n=s\frac{q-1}{u}+t\frac{q-1}{v}-\frac{2(q-1)\gcd(u,v)}{uv}st$, \\
$u\mid v(r+1)$, $v\mid u(r-1)$, \\$n$, $u$, $\frac{(r+1)v}{u}s+v$ even,\\
$1\leq s\leq \frac{u}{\gcd(u,v)}$, $1\leq t\leq \frac{v}{\gcd(u,v)}$} & \cite{FXF21}\\
\hline
38&$q=r^{2}$, $r \equiv 3 \pmod 4$ & \tabincell{c}{$n=s\frac{q-1}{u}+t\frac{q-1}{v}-\frac{2(q-1)\gcd(u,v)}{uv}st+1$,\\ 
$n$ even,  $u\mid v(r+1)$, 
$v\mid u(r-1)$,\\ $1\leq s\leq \frac{u}{\gcd(u,v)}$, $1\leq t\leq \frac{v}{\gcd(u,v)}$, \\$u$, $\frac{(r+1)v}{u}s+v$ and $\frac{v(r+1)}{u}\frac{s(s-1)}{2}$ even} & \cite{FXF21}\\
\hline
39 & $q=r^{2}$, $r \equiv 3 \pmod 4$ & \tabincell{c}{$n=s\frac{q-1}{u}+t\frac{q-1}{v}-\frac{2(q-1)\gcd(u,v)}{uv}st+2$, \\$n$ even,
$u\mid v(r+1)$, $v\mid u(r-1)$, \\ $\frac{(r+1)v}{u}s$ and $\frac{v(r+1)}{u}\frac{s(s-1)}{2}$ even,\\
$1\leq s\leq \frac{u}{\gcd(u,v)}$, $1\leq t\leq \frac{v}{\gcd(u,v)}$} & \cite{FXF21}\\
\hline
40&$q=r^2$ & \tabincell{c}{$n=s\frac{q-1}{e_1}+t\frac{q-1}{e_2}$, $e_1 \equiv 2^l \pmod  {2^{l+1}}$, \\$2^l\mid e_2$, where $l\geq 2$, $4\mid (s-1)(r+1)$, \\ $2 e_2\mid e_1(r-1), e_1\mid e_2(r+1)$,\\  $1\leq s\leq \frac{e_1}{\gcd(e_1,e_2)}$, $1\leq t\leq \frac{e_2}{\gcd(e_1,e_2)}$    } &\cite{WLZ23}\\
\hline
41& $q=r^2$ & \tabincell{c}{$n=s\frac{q-1}{e_1}+t\frac{q-1}{e_2}+1$, $e_1 \equiv 2^l \pmod {2^{l+1}}$, \\$2^l\mid e_2$, where $l\geq 2$, $4\mid (s-1)(r+1)$,\\ $2e_2\mid e_1(r-1), e_1\mid e_2(r+1)$,  \\$1\leq s\leq \frac{e_1}{\gcd(e_1,e_2)}$, $1\leq t\leq \frac{e_2}{\gcd(e_1,e_2)}$    } &\cite{WLZ23}\\
\hline
42& $q=r^2$ & \tabincell{c}{$n=s\frac{q-1}{e_1}+t\frac{q-1}{e_2}+2$, $e_1 \equiv 2^l \pmod  {2^{l+1}}$, \\$2^l\mid e_2$, where $l\geq 2$, $4\mid (s-1)(r+1)$,\\ $2e_2\mid e_1(r-1), e_1\mid e_2(r+1)$, \\ $1\leq s\leq \frac{e_1}{\gcd(e_1,e_2)}$, $1\leq t\leq \frac{e_2}{\gcd(e_1,e_2)}$    } &\cite{WLZ23}\\
\hline
43&$q=r^2$ & \tabincell{c}{$n=s\frac{q-1}{e_1}+t\frac{q-1}{e_2}$, $e_1 \equiv 2^l \pmod  {2^{l+1}}$, \\$2^l\mid e_2$, where $l\geq 2$, $\frac{r+1}{2}(\frac{te_1}{e_2}+1)$ even, \\$2e_2\mid e_1(r+1), e_1\mid e_2(r-1)$, \\$1\leq s\leq \frac{e_1}{\gcd(e_1,e_2)}$, $1\leq t\leq \frac{e_2}{\gcd(e_1,e_2)}$    } &\cite{WLZ23}\\
\hline
44&$q=r^2$ & \tabincell{c}{$n=s\frac{q-1}{e_1}+t\frac{q-1}{e_2}+1$, $e_1 \equiv 2^l \pmod  {2^{l+1}}$, \\$2^l\mid e_2$, where $l\geq 2$, $\frac{r+1}{2}(\frac{te_1}{e_2}+t)$ even,\\ $4\mid (s-1)(r+1)$, $2e_2\mid e_1(r+1)$, $e_1\mid e_2(r-1)$,\\  $1\leq s\leq \frac{e_1}{\gcd(e_1,e_2)}$, $1\leq t\leq \frac{e_2}{\gcd(e_1,e_2)}$    } &\cite{WLZ23}\\
\hline
45&$q=r^2$ & \tabincell{c}{$n=s\frac{q-1}{e_1}+t\frac{q-1}{e_2}+2$, $e_1 \equiv 2^l \pmod { 2^{l+1}}$, \\$2^l\mid e_2$, where $l\geq 2$, $\frac{r+1}{2}(\frac{te_1}{e_2}+t)$ even,\\ $4\mid (s-1)(r+1)$, $2e_2\mid e_1(r+1)$,  $e_1\mid e_2(r-1)$,\\ $1\leq s\leq \frac{e_1}{\gcd(e_1,e_2)}$, $1\leq t\leq \frac{e_2}{\gcd(e_1,e_2)}$    } &\cite{WLZ23}\\
\hline
\end{longtable}
}
\section*{Ackonwledgements}
This research was supported in part by the National Key Research and Development Program of China under Grant Nos. 2022YFA1004900, 2021YFA1001000 and 2022YFA1005000, the National Natural Science Foundation of China under Grant Nos. 62201322, 12141108, 62371259, 12471493 and 12441105,  the Natural Science Foundation of Shandong (ZR2022QA031), the Taishan Scholar Program of Shandong Province, the Natural Science Foundation of Tianjin (20JCZDJC00610), the Fundamental Research Funds for the Central Universities of China (Nankai University) and the Nankai Zhide Foundation.

\bibliographystyle{model1a-num-names}

\end{document}